%% file: BackupLongVersion.tex
\documentclass{lmcs} %%% last changed 2014-08-20


\usepackage{hyperref}
\usepackage{mathrsfs}
\usepackage{soul}
\theoremstyle{plain} %\crefname{satz}{Satz}{S\"atze}
\usepackage{float}
\usepackage{graphicx}
\usepackage[all,2cell]{xy}
\UseAllTwocells
\xyoption{v2}
\NoResizing
\UseResizing
\usepackage{bussproofs}

\newcommand{\red}{\textcolor{red}}

\newcommand{\blue}{\textcolor{blue}}
\newcommand{\commento}[1]{}

\newcommand{\ewire}{\textsf{EWire}}
\newcommand{\qwire}{\textsf{QWire}}

\newcommand{\host}{\mathsf{H}}
\newcommand{\core}{\mathsf{C}}
\newtheorem{remark}{Remark}
\newtheorem{lemma}{Lemma}
\newcommand{\tenscont}[1]{\times_{#1}}
\newcommand{\ptype}{\mathtt{Proof}}
\def\REL{\mathbf{Rel}}
\def\DCPO{\mathbf{Dcpo}}

\newcommand{\lang}[1]{\mathscr{L}_{#1}}
\newcommand{\ctypeset}{\mathscr{T}_{C}}
\newcommand{\htypeset}{\mathscr{T}_{H}}

\newcommand{\boldbit}{\mathbf{Bit}}

\newcommand{\boldbool}{\mathsf{bool}}

\newenvironment{scprooftree}[1]%
  {\gdef\scalefactor{#1}\begin{center}\proofSkipAmount \leavevmode}%
    {\scalebox{\scalefactor}{\DisplayProof}\proofSkipAmount \end{center} }
  \newcommand{\enrichedparallel}[2]{\mathbf{parallel}(#1,#2)}

\newcommand{\tenscontlin}[1]{\otimes_{#1}}
\newcommand{\LNL}{\texttt{LNL}}

\begin{document}

\include{macros}

%\title[Compositional theories for embedded languages]{Compositional theories for embedded \blue{??host-core???} languages} 

\title[Compositional Theories for Host-Core Languages]{Compositional Theories for Host-Core Languages}
%\titlecomment{{\lsuper*}OPTIONAL comment concerning the title, \eg, 
%  if a variant or an extended abstract of the paper has appeared elsewhere.}

%\thanks{thanks 1, optional.}	%optional
\author[D. Trotta]{Davide Trotta}	%required
\address{University of Pisa}	%required
\email{}  %optional
%\thanks{thanks 1, optional.}	%optional

\author[M.~Zorzi]{Margherita Zorzi}	%optional
\address{University of Verona, Italy}	%optional
\email{margherita.zorzi@univr.it}  %optional
%\thanks{thanks 2, optional.}	%optional

%\author[C.~Name3]{Carla Name3}	%optional
%\address{address 3}	%optional
%\urladdr{name3@url3\quad\rm{(optionally, a web-page can be specified)}}  %optional
%\thanks{thanks 3, optional.}	%optional

%% etc.

%% required for running head on odd and even pages, use suitable
%% abbreviations in case of long titles and many authors:

%%%%%%%%%%%%%%%%%%%%%%%%%%%%%%%%%%%%%%%%%%%%%%%%%%%%%%%%%%%%%%%%%%%%%%%%%%%

%% the abstract has to PRECEDE the command \maketitle:
%% be sure not to issue the \maketitle command twice!

\begin{abstract}

 Linear type theories, of various types and kinds, are of fundamental importance in most programming language research nowadays. In this paper we describe an extension of Benton's Linear-Non-Linear type theory and model for which we can prove some extra properties that make the system better behaved as far as its theory is concerned. We call this system the host-core type theory. The syntax of a host-core language is {split into} two parts, representing respectively a host language $\mathsf{H}$ and a core language $\mathsf{C}$, embedded in $\mathsf{H}$. This idea, derived from  Benton’s Linear-Non-Linear formulation of Linear Logic, allows a flexible management of data linearity, which is particularly useful in non-classical computational paradigms. The host-core style can be viewed as a simplified notion of multi-language programming, the process of software development in a heterogeneous programming language. In this paper, we present the typed calculus $\VCcalcolo$, a minimal and flexible host-core system that captures and standardizes common properties of an ideal class of host-core languages. We provide a denotational model in terms of enriched categories and we state a strong correspondence between syntax and semantics through the notion of internal language. The latter result provides some useful characterizations of host-core style, otherwise difficult to obtain. We also discuss some concrete instances, extensions and specializations of the system $\VCcalcolo$.

\end{abstract}

\maketitle

%% start the paper here:
\section{Introduction}

{
The idea of dividing the syntax of a formal system into two communicating parts was developed both in logic and in computer science. In ~\cite{AMLNLLPTM} Benton introduced \LNL\ (Linear-Non-Linear Logic), a presentation of Linear Logic where the bang (!) modality is decomposed into an adjunction  between a symmetric monoidal closed category and a cartesian closed category.
\LNL\ models are a considerable simplification of those of Intuitionistic Linear Logic and, in the following years, \LNL\ catalysed the attention of the categorical logic community. Maietti et al. in ~\cite{RCSILL} discussed models and morphisms for several versions of Linear Logic, including \LNL.  Benton’s idea enjoys an interesting interpretation also from a programming language perspective. In fact, the recent quantum computing literature shows that \LNL\ provides a good foundation for paradigmatic calculi in which a ``main language” controls the computation and delegates the management of quantum data to a linear embedded ``core”. \LNL\ has inspired interesting working quantum circuit definition languages such as $\qwire$~\cite{qwire17a,qwire17b,RobRandThesis,Paykin18} and its generalization $\mathsf{EWire}$, whose foundation and denotation have been partially studied by Staton et al. ~\cite{RS20}.

A situation where a main language interacts with others is not a peculiarity of quantum languages. In practical programming, it is frequent to use a principal encoding language and to distribute some tasks to another (often domain specific) language, ``called” or imported in terms of ready-to-use subroutines. Programming platforms and modern IDE systems provide efficient methods to manage and control the process of software development in a heterogeneous programming language environment and software developers need tools to combine different languages, sometimes called \emph{cross-language interoperability mechanisms} ~\cite{Chisnall13}. Multi-Language Programs, as named in a pioneering investigation by Matthews and Findler ~\cite{MatthewsF07}, became a flourishing research area in programming languages. Even if some interesting results have been achieved in the last decade ~\cite{OseraSZ12,PattersonPDA17,BuroM19,BuroCM20} this is still an open research field, in particular from the programming foundation perspective. 

Following the intuitions described above, we propose a typed calculus called $\VCcalcolo$, built upon two communicating languages $\mathsf{H}$ and $\mathsf{C}$. We design and study $\VCcalcolo$ following two different traditions. On the one hand, we follow the categorical logic tradition started by Benton, and for that \cite{RCSILL} is our main technical reference. On the other hand, (quantum) circuit description languages such as $\qwire$ and $\ewire$ (and also their ancestor $\mathsf{Quipper}$) are the kind of application we envisage: in particular, we adopt from  ~\cite{RS20} the use of enriched categories as our denotational model. Beyond these technical references, we take some inspiration from practical multi-language programming. To better frame our work with respect to the state of art, we identify the notion of a \emph{host-core language} as a simplification of the notion of a \emph{multi-language framework}. 

The syntax of a host-core language is split into two parts, one for the \emph{host}, an arbitrarily powerful language $\host$, and one for the \emph{core} language $\core$, which is \emph{embedded} in $\host$. Their interaction is controlled through mixed typing rules that specify how to ``promote” well-formed terms from the core language to the host language and (possibly) conversely. With respect to a more general definition of multi-language framework (coming from \cite{OseraSZ12,PattersonPDA17,MatthewsF07, BuroM19,BuroCM20}): 
\begin{enumerate}
  \item we focus our attention on exactly two languages; 
  \item we model a restricted form of interoperability, in which the host can import all core terms (and types) but the core cannot import any terms from the host. This could be seen as a limitation, but we  show how even this restricted case is challenging enough to model, and  already provides interesting results. Notice that Benton’s \LNL\ type theory, which represents our starting point, consists of the particular case of a host-core language where the core is strictly linear.
\end{enumerate}
The ultimate goal of $\VCcalcolo$ (and its future developments) is to define a completely compositional theory that allows one to reason about the host language, the core language, and their communication, pursuing  full flexibility of the modelling framework. In this first investigation we propose a study of what we call a ``kernel” type theory, common to an ideal class of host-core languages. We propose a semantics for these host-core languages in terms of \emph{enriched category theory}, obtaining an \emph{internal language theorem} (in the sense of categorical logic), which states a full correspondence between the syntax and semantics of $\VCcalcolo$ type theories.
}

 \subsection{Main Results}\label{sec:contribution}
 \commento{
\blue{In this paper we carry on the research directions started by Benton and Staton, focusing on type theories for (a definition of) host-core languages, their denotational models in terms of enriched categories and establishing a strong correspondence between syntax and semantics via the notion of \emph{internal language}. Moreover, we show the notion of internal language  answers some central questions about the theoretical foundation of (a notion of) host-core programming languages.} \blue{POSSIAMO OMETTERE QUESTO CAPPELLO}  

We list our main results.
}

\begin{description}
\item[1) A ``kernel'' type theory for host-core languages] {We describe a kernel type theory for host-core languages where the core language is linear and the host language includes the usual lambda-calculus: we define a systematised type theory, which formally defines this notion of host-core languages. We address the problem from a general  algebraic perspective and we present a minimal system called $\VCcalcolo$, a simply-typed lambda calculus hosting a purely linear core language. Our starting point is a simple language but our goal is not only to achieve a Linear-Non-Linear system, but also to be able to easily extend both  the host and  the core fragments of the system. We  choose to start from a linear core, since, as we will show, we can obtain a non-linear core as a particular instance. To this end we exploit the $\VCcalcolo$ compositionality.
The design of the $\VCcalcolo$ type system “privileges” the host language $\host$, {establishing} a hierarchical dependency of the core fragment on the control fragment: the embedded language $\core$ is fully described in the host $\host$.} 
%
%\blue{QUI HO RIFORMULATO E IPER-SEMPLIFICATO. CONTROLLARE ALLA FINE}
 {%Then, we show how to extend the core language $\core$ to a non-linear one. 
 %To this end, we exploit the $\VCcalcolo$ compositionality. 
 We will provide a running example that explains the design choices we made, showing $\VCcalcolo$ instances and extensions.
}  
 %Differently from other proposals, we show how we can easily break the dichotomy LNL (inherited from Benton's tradition) allowing to describe an \emph{arbitrary core} into an \emph{arbitrary host}, preserving the correctness of the system.  To this end, we exploit $\VCcalcolo$ compositionality, providing an example of extension of the type theory and the corresponding semantics denotation.
% We also provide a  running example that explains the design choices we made and shows $\VCcalcolo$ instances and extensions.

\item[2) Denotation of  host-core calculi] 
{We provide a semantics for $\VCcalcolo$ in terms of enriched categories, partially following the blueprint of ~\cite{RennelaS19}. The compositionality we claim for the $\VCcalcolo$ syntax is mirrored at the semantic level. The models of $\VCcalcolo$ are pairs $(\mV,\mC)$ where $\mV$ is a cartesian closed category and $\mC$ is a $\mV$-symmetric monoidal category. In particular, we show how the presentation of the calculus admits a natural correspondence between syntactic properties and semantics constraints: $\beta$ and $\eta$ rules become coherence conditions on the morphisms of the algebraic structures. Ideally, we want to pursue a full correspondence between syntax and $\VCcalcolo$ models in order to obtain more refined type theories by simply adding syntactic rules and (equivalent) denotational properties, without changing the rules of the basic language $\VCcalcolo$.
}

 \item[3) Internal language for host-core calculi] 
 
 {We study the relationship between host-core languages and their models, using $\VCcalcolo$ as a case study. {We point out that the expected soundness and completeness theorems may be not enough to fully identify the most appropriate class of models for a type theory}, as explained in the work~\cite{RCSILL}. Moving beyond soundness and completeness, we relate host-core languages and their categorical models via the notion of \emph{internal language}. One can find several examples and applications of the notion of internal language in the categorical logic literature,  e.g. see ~\cite{IHOCL,CLP,SAE}. 
}{
By definition, $\VCcalcolo$ provides an internal language of the category 
$\modelcat{\VCcalcolo}$ of its models if one proves, as we do in Section \ref{sec:semantics}, an equivalence between $\modelcat{\VCcalcolo}$ and the category $\catvcalcolo$ of the theories of the language. We remark that in this paper we are working in the context of enriched categories and the definition of a suitable notion of morphisms for the category of models $\modelcat{\VCcalcolo}$ is the interesting and challenging point of the proof. This requires the technical notion of ``change of base'' which formally describes how, changing the host language, one can induce a change of the embedding for a given core language. After we show the correspondence between syntax and semantics via an internal language theorem, we provide some concrete examples of models.
}

\item[4) A bridge towards host-core programming theory] 
  {
 %\red{We propose a foundational investigation by focusing on linear/non-linear type theory and categorical semantics.} 
 In the last part of the paper we show how the results we proved can be used to provide an initial but significant formal definition of host-core languages. In the last few years  the foundations of multi-language programming theory have been much discussed~\cite{OseraSZ12,PattersonPDA17,BuroM19,BuroCM20}, however a key question remains  unexplored: is there a formal characterization of a  programming theory, as a type theory and a denotational semantics, where different heterogeneous languages are used and allowed to communicate with each other?
}

{
Focusing on the “host-core” languages as a particular simplified case of the most general notion of multi-languages, we answer three instances of the generic questions below. Consider two programming languages $\lang{1}$ and $\lang{2}$,  supplied with their own syntax, semantics and mathematical properties.
\begin{description}
\item[Q1] Is there a language $\lang{3}$ built out from $\lang{1}$ and $\lang{2}$, that describes both $\lang{1}$ and $\lang{2}$ and their communication allowing for the import and export of terms from one syntax into the other? If so, what is $\lang{3}$'s formal semantics? 
\item[Q2] If $\lang{1}$ is a host-core language, can we compare $\lang{1}$ with another standalone language ${\lang{2}}$?

\item[Q3] Does the internal language result represent a useful tool for designing new host-core languages?
 \end{description}
The above questions are particularly relevant for host-core languages and are non-trivial, since $\lang{1}$ and $\lang{2}$ could be radically different. Moreover, we do not consider mere extensions of a host language. We try to understand when a given language is expressive enough to host another given language, to delegate specific computations to it, and to represent its programs at the host syntax level.
}
%
%\red{Nota su Q2: indebolirei questa domanda, secondo quanto fatto con SyntaxGen. Che ne dici di qualcosa tipo Q2: Given $\lang{1}$ and $\lang{2}$ specification,  there exists a method which guides the design of a type theory for a host-core language $\lang{3}$ such that  $\lang{1}$ hosts $\lang{2}$?}

%\blue{RISCRITTO QUANTO SEGUE  togliendo la giustapposizione, LEGGERE  CON ATTENZIONE,}
%Thus,  a statement like ``\emph{$\lang{1}$ hosts $\lang{2}$}'' (which encodes question Q1) is hard to formally define syntactically. In fact, one should say that there exists a language $\lang{3}$ built by mixing  $\lang{1}$ and  $\lang{2}$ syntaxes and type theories where $\lang{1}$ hosts a core whose ``pure'' part is  $\lang{2}$. 
%\blue{FINE PARTE}

As shown in Section ~\ref{sec:travail}, we achieve this “semantically”. Thanks to the existence of the denotational model and the strong correspondence with the syntax, we can answer the question \textbf{Q1}, providing a definition of the host-core language $\lang{3}$. As far as  \textbf{Q2} is concerned, notice that the problem of formally comparing a standalone language with a host-core language is a very tricky task. Once again, the internal language theorem gives a preliminary solution.
 The correspondence provides an initial answer also to the question \textbf{Q3}. 
\end{description}

\bemph{Synopsis.}
This paper is organized as follows. In Section~\ref{sec:calculus} we present the syntax, type system and evaluation rules of the calculus $\VCcalcolo$, as well as a running example. In Section ~\ref{sec:semantics} we describe the categorical models of $\VCcalcolo$ and we prove the correspondence between syntax and denotation via the notion of internal language. In Section ~\ref{sec:travail} we use the notion of internal language to answer the questions above. In Section ~\ref{sec:related} we recall the state of the art on multi-language frameworks; discussions, conclusions and future work are in Section ~\ref{sec:conclusions}. Finally, in Appendix~\ref{app:evaluation rules} we show the evaluation rules of the host languages $\host$ and in Appendix ~\ref{sec:appedix} we recall some background notions about enriched categories.

\section{A ``kernel'' host-core calculus}\label{sec:calculus}

%\blue{Seguendo il revisore, fissiamo la terminologia. Usiamo sempre ``host-core langages/programming'' quando parliamo dei nostri sistemi, anche per individuare quello che facciamo: un caso particolare e approfondito di un paradigma pi\`u generale. Lo diremo nell'intro e nella lettera ai revisori}.
In this section we present the typed calculus $\VCcalcolo$,  %We have two kinds of different types, terms and judgements: 
%Moreover we present a categorical semantic using enriched categories, and we show using the notion internal language generalized in the enriched case, that the category of theory over the language $\VCcalcolo$ is equivalent to the category whose objects are enriched categories with some properties.
%For the rest of this section we consider $\mV$ be a locally small cartesian closed category, and consider a $\mV$-symmetric monoidal enriched category $\mC$. 
 built upon two components: a \emph{host language} $\host$ designed as a simply typed lambda calculus, and linear \emph{embedded core language} $\core$.
  
%The idea to start with two independent typed calculi $\host$ and $\core$ and adding operations that allow the communication between the systems is one of the main advantages of the presentation of the LNL-logic introduced by Benton \cite{AMLNLLPTM}. 

The presentation of the typed calculus proposed here gives a \emph{privileged} position to the host language, in the sense that the embedded language is fully described in the host one. This means that, differently from Benton's logic, in $\VCcalcolo$ the communication between the languages $\host$ and $\core$ is not possible in both  directions. Roughly speaking,  we cannot carry all the terms of $\host$ into the language $\core$, modelling a hierarchical dependency of the core on the host language. 
From the point of view of the categorical semantics  we will not have any comonad or adjunction between the categories that provide denotation to the host and the core components but, as we will see in Section~\ref{sec:semantics}, we have just an instance of an \emph{enriched category}.

 \subsection{General presentation of the typed system}
 
% \blue{LO SPOSTEREI O LO ALLEGGERIREI
%In the following, we mainly refer to \cite{RCSILL} and \cite{AMLNLLPTM} for syntactical notation and categorical semantic of the LNL-typed calculus, while for more general introduction to the categorical semantic we refer to \cite{IHOCL}, \cite{CLP}, \cite{CLTT}.
%}

%As the LNL-system,  $\VCcalcolo$  presents two kinds of different types, terms and judgements: $\mV$-types and $\mV$-terms for the host language, and $\mC$-types and $\mC$-terms for the pure linear language embedded in $\mV$. 

We introduce type constructors, well-formation and evaluation rules for $\VCcalcolo$.  
Notice that we describe the rules in logical style, {assuming that we have a
universe of (base) types.}
Concrete instances of $\VCcalcolo$ can be easily obtained setting base types to a set of interest and building other types by type constructors. A concrete instance of type and term syntax  is in our running example (see Examples~\ref{ex:running2}, \ref{ex:running3} and \ref{ex:running4} below). We strictly follow the notation used in~\cite{RCSILL} and  we consider two different kinds of types, terms-in-contexts and judgements: $\host$-types and -terms for the host language $\host$ and $\core$-types and -terms, for the pure linear language $\core$ which is embedded in $\host$.

Type and term judgements and type and term equalities for $\host$ are the following ones
\begin{table}[H]
\centering
\renewcommand\arraystretch{2.5}
\begin{tabular}{l l l l }
$ \typeconstructor{\host}{X}$& $\typeconstructor{\host}{X=Y}$& $\termconstructor{\Gamma}{\host}{t}{X}$ &  $\termconstructor{\Gamma}{\host}{t=s}{X.}$
\end{tabular}\label{tab:vjud}
\end{table}
\noindent
The first judgement says that a syntactical object is of sort type (in the empty context), the second  that two types are equal, the third states that a term $t$ in a context $\Gamma$ has type $X$ and the fourth that two terms are equal. \commento{Notice we are only considering  simple types, i.e. types without dependencies.}
Type and term judgements with corresponding equalities for the core language $\core$ are of the following form
 
\begin{table}[H]
\centering
\renewcommand\arraystretch{2.5}
\begin{tabular}{l l l l }
$ \typeconstructor{\core}{A}$& $\typeconstructor{\core}{A=B}$& $\terconstrmix{\Gamma}{\Omega}{\core}{f}{A}$ &  $\terconstrmix{\Gamma}{\Omega}{\core}{f=h}{A}$
\end{tabular}\label{tab:cjud}
\end{table}
\noindent
where ${\Gamma}\;|\;{\Omega}{}{}$ denotes a \emph{mixed context}. In a mixed context  we split out the host part $\Gamma$ and the core part $\Omega$. Notice that, as a particular case, both  $\Gamma$ and  $\Omega$ can be empty.

Given a host context $\Gamma$ (resp. a core context $\Omega$) we denote as $\tenscont{\Gamma}$ (resp.  $\tenscontlin{\Omega}$) the  product of the elements in  $\Gamma$ (resp. in $\Omega$).

To stress the distinction between the host and the core parts of the syntax, we write $X,Y,Z\dots$ for the $\host$-types, $\Gamma,\Gamma'$ for $\host$-contexts, $x,y,z\dots$ for $\host$-variables, and $s,t,v\dots$ for $\host$-terms. We write $A,B,C\dots$ for the $\core$-types, $\Omega,\Omega'$ for $\core$-contexts, $a,b,c\dots$ for $\core$-variables, and $h,f,g\dots$ for $\core$-terms.

\bigskip

 We consider two fixed sets of \emph{base types}:  
 %a set of base $\core $-types, denoted by $\alpha$ (possibly indexed) and a set of $\host$-base types, denoted by $\iota$ (possibly indexed). 
  a set $\Sigma_\core$ of base $\core$-types, ranged over by $\alpha$ (possibly
 indexed) and a set $\Sigma_\host$ of base H-types, ranged over by $\iota$
 (possibly indexed).

 Similarly to \cite{qwire17a,RS20}, we have a host-type of the form $\Ctype{\ptype}{A}{B}$ (where  $A$ and $B$ are core-types). The proof-theoretical meaning of a type $\Ctype{\ptype}{A}{B}$ is that it represents the type of proofs from $A$ to $B$ in the host language. Given a core term, we can read it as a proof and in this way we can speak about $\core$-proofs at the host level.

The type constructor rules are presented in Figure~\ref{fig:typeconstr}.
Notice that in $(t2c)$ we obtain a $\host$-type starting from two $\core$-types. 
%
%\blue{io qui introdurrei una spiegazione del tipo PROOF.
%CIRC diverso da esponente, entrambi rappresentano funzioni, possiamo comporre ma C non ha la valutazione, per valutare dobbiamo entrare in host}

\begin{figure}[H]
\centering
\renewcommand\arraystretch{2}
%\small
\begin{tabular}{l  c }
 $(t01c)$ &  $\typeconstructor{\core}{\alpha}$ for any base $\core$-type $\alpha$\\
$(t02c)$ & $\typeconstructor{\core}{I}$\\
 $(t1c)$ & \AxiomC{$\typeconstructor{\core}{A}$}
\AxiomC{$\typeconstructor{\core}{B}$}
\BinaryInfC{$\typeconstructor{\core}{A\otimes  B}$}
\DisplayProof\\
$(t2c)$ & \AxiomC{$\typeconstructor{\core}{A}$}
\AxiomC{$\typeconstructor{\core}{B}$}
\BinaryInfC{$\typeconstructor{\host}{\Ctype{\ptype}{A}{B}}$}
\DisplayProof\\
$(t01h)$ &  $\typeconstructor{\host}{\iota}$ for any base $\host$-type $\iota$\\
$(t02h)$ & $\typeconstructor{\host}{1}$\\
$(t1h)$ & \AxiomC{$\typeconstructor{\host}{X}$}
\AxiomC{$\typeconstructor{\host}{Y}$}
\BinaryInfC{$\typeconstructor{\host}{X\times Y}$}
\DisplayProof\\
$(t2h)$ & \AxiomC{$\typeconstructor{\host}{X}$}
\AxiomC{$\typeconstructor{\host}{Y}$}
\BinaryInfC{$\typeconstructor{\host}{\arrowtype{X}{Y}}$}
\DisplayProof\\
\end{tabular}
\caption{Type constructors}\label{fig:typeconstr}
\end{figure}

We require that all the type constructors preserve the equality of types as usual. For example, we require $\Ctype{\ptype}{-}{-}$ to preserve the equality of types, i.e. we have to add the rule

\begin{prooftree}
$(eqT)$\AxiomC{$\typeconstructor{\core}{A=B}$}
\AxiomC{$\typeconstructor{\core}{C=D}$}
\BinaryInfC{$\typeconstructor{\host}{\Ctype{\ptype}{A}{C}=\Ctype{\ptype}{B}{D}}$}
\end{prooftree}
to the system.

%The grammars of terms are those of a simply-typed lambda calculus and a linear calculus with $\mathsf{let}$ constructors for the host $\host$ and the core $\core$, respectively. % We also introduce the host unit term $\ast$ with its host unit type $1$ and the core unit term $\bullet$ with its core unit type $I$.

%In the following example we introduce an instance of  $\VCcalcolo$, and we call it  
%$\vcinstance$. We choose the natural syntax induced by type constructors and equations defined above.  Although this step is clearly theoretically negligible, it is useful for the rest of the section. On the basis of $\vcinstance$, we are going to develop some incremental examples to show some features and properties of  $\VCcalcolo$ throughout the section. 

%
%
%\blue{NOTA: potremmo mettere l'istanza anche fuori dall'esempio, intanto la lascio qui. Secondo me, anche se non \`e il massimo lasciarla nell'esempio, pu\`o essere utile per sottolineare che NON stiamo scrivendo un articolo di linguaggi.}

Let $\ctypeset$ be the set of core types.  The grammar of core types is defined as:
$$
A,B,C~::=~ I~ |~ \alpha ~|~ A\otimes B
$$
where $I$ is the unit type  and $\alpha\in\Sigma_{\core}$.% with $\Sigma_{\core}$ collection of base types. 

Let $\htypeset$ the set of host types. The grammar of host types is defined as 
\newcommand{\intt}{\mathsf{Int}}
\newcommand{\boolt}{\mathsf{Bool}}
$$
X,Y,Z~::=~ 1  ~|~ \iota ~|~X\times Y~|~X\rightarrow Y~|~\ptype(A,B)
$$
where $1$ is the unit type, $\iota\in\Sigma_{\host}$ and $A$ and $B$ in $\ptype{}{}$ are $\core$-types. $\Sigma_{\host}$ could be instantiated, for example, to the set $\{\mathsf{Nat},\mathsf{Bool}\}$.
The type $\ptype(A,B)$ represents and models the meeting point between the two syntaxes.

%The proof-theoretical intuition is that a type $\Ctype{\ptype}{A}{B}$ represents the type of \emph{proofs} from $A$ to $B$ in the host language. Given a core term, we can read it as a proof and in this way we can represent 
%$\core$-proofs at the host level.

%
The grammar of $\core$ terms is
$$
f,g::= \bullet\;|\; a\;|\; f\otimes g\;|\; \letbein{f}{a\otimes b}{h} \;|\; \letbein{f}{\bullet}{h}\;|\;\derelict{s,f}
$$
where $s$ is a host term, i.e. it ranges over the following  $\host$ grammar: 
$$s,t::=\ast\;|\;x\;|\;\angbr{s}{t}\;|\;\inl{s}\;|\;\inr{s}\;|\;\lambda x:X.t\;|\;st\;|\; \promote{a_1,\ldots,a_n.f}$$
with $f$ a core term. 

{ The grammars of terms are those of a simply-typed lambda calculus and a linear calculus with $\mathsf{let}$ constructors for the host $\host$ and the core $\core$, respectively.  We also introduce the host unit term $\ast$ of host unit type $1$, and the core unit term $\bullet$ of core unit type $I$.
Moreover, we add to the host and the core grammars respectively the functions  $\promote{\cdot.\cdot}$ and  $\derelict{\cdot,\cdot}$  that model the \emph{communication} between the host and the core syntaxes.}

\bigskip

The typing rules for $\VCcalcolo$ are in Figure~\ref{fig:typingrules}. 

\noindent 
%\begin{notation}
%{\textbf{Notation}. Given a $\core$-context $\Omega:=[a_1:A_1,\dots,a_n:A_n]$ we use the notion $\otimes_{\Omega}$ to indicate the $\core$-type $\otimes_{\Omega}:=A_1\otimes \cdots \otimes A_n$. }
%\end{notation}

\begin{figure}[H]
%\begin{figure}
\centering
\renewcommand\arraystretch{3}
\small
\begin{tabular}{l  c}
(av)\qquad$\termconstructor{\Gamma_1,x:X,\Gamma_2}{\host}{x}{X}$& (uv)\qquad$\termconstructor{\Gamma}{\host}{\ast}{1}$ \\

 (pv)\AxiomC{$\termconstructor{\Gamma}{\host}{s}{X}$}
\AxiomC{$\termconstructor{\Gamma}{\host}{t}{Y}$}
\BinaryInfC{$\termconstructor{\Gamma}{\host}{\angbr{s}{t}}{X\times Y}$}
\DisplayProof &   ($\pi 1v$) \AxiomC{$\termconstructor{\Gamma}{\host}{v}{X\times Y}$}
\UnaryInfC{$\termconstructor{\Gamma}{\host}{\inl{v}}{X} $}
\DisplayProof\\
 ($\pi 2v$)\AxiomC{$\termconstructor{\Gamma}{\host}{v}{X\times Y}$}
\UnaryInfC{$\termconstructor{\Gamma}{\host}{\inr{v}}{Y} $}
\DisplayProof & (aiv)\AxiomC{$\termconstructor{\Gamma,x:X}{\host}{t}{Y}$}
\UnaryInfC{$\termconstructor{\Gamma}{\host}{\lambda x:X.t}{\arrowtype{X}{Y}}$}
\DisplayProof\\
(aev)\AxiomC{$\termconstructor{\Gamma}{\host}{t}{\arrowtype{X}{Y}}$}
\AxiomC{$\termconstructor{\Gamma}{\host}{s}{X}$}
\BinaryInfC{$\termconstructor{\Gamma}{\host}{t(s)}{Y}$}
\DisplayProof & \AxiomC{$\termconstructor{\Gamma}{\host}{t}{\Ctype{\ptype}{A}{B}}$}\\
%\AxiomC{$\termconstructor{\Gamma}{\mV}{s}{\Ctype{\mC}{B}{C}}$}
%\BinaryInfC{$\termconstructor{\Gamma}{\mV}{\enrichedcomp{t}{s}}{\Ctype{\mC}{A}{C}}$}
%\DisplayProof \\
%$\termconstructor{\Gamma}{\mV}{\id_A}{\Ctype{\mC}{A}{A}}$

(ac)\qquad$\terconstrmix{\Gamma}{a:A}{\core}{a}{A}$ & (uc)\qquad$\terconstrmix{\Gamma}{-}{\core}{\bullet}{I}$\\
(tc)\AxiomC{$\terconstrmix{\Gamma}{\Omega_1}{\core}{f}{A}$}
\AxiomC{$\terconstrmix{\Gamma}{\Omega_2}{\core}{g}{B}$}
\BinaryInfC{$\terconstrmix{\Gamma}{\Omega_1,\Omega_2}{\core}{f\otimes g}{A\otimes B}$}
\DisplayProof &\\
(let1c)\AxiomC{$\terconstrmix{\Gamma}{\Omega_1}{\core}{f}{A\otimes B} $}
\AxiomC{$\terconstrmix{\Gamma}{\Omega_2,a:A,b:B}{\core}{h}{C}$}
\BinaryInfC{$\terconstrmix{\Gamma}{\Omega_1,\Omega_2}{\core}{\letbein{f}{a\otimes b}{h}}{C}$}
\DisplayProof &\\
(let2c)\AxiomC{$\terconstrmix{\Gamma}{\Omega_1}{\core}{f}{I} $}
\AxiomC{$\terconstrmix{\Gamma}{\Omega_2}{\core}{h}{C}$}
\BinaryInfC{$\terconstrmix{\Gamma}{\Omega_1,\Omega_2}{\core}{\letbein{f}{\bullet}{h}}{C}$}
\DisplayProof &
\\
(prom)\AxiomC{$\terconstrmix{\Gamma}{\Omega}{\core}{f}{A}$}
\UnaryInfC{$\termconstructor{\Gamma}{\host}{\promote{a_1,\ldots,a_n.f}}{\Ctype{\ptype}{A_1\otimes\ldots\otimes A_n}{A}}$}
\DisplayProof 

& \mbox{where} $\Omega=[a_1:A_1,\ldots,a_n:A_n]$\\
(der)\AxiomC{$\termconstructor{\Gamma}{\host}{t}{\Ctype{\ptype}{A}{B}}$}
\AxiomC{$\terconstrmix{\Gamma}{\Omega}{\core}{f}{A}$}
\BinaryInfC{$\terconstrmix{\Gamma}{\Omega}{\core}{\derelict{t,f}}{B}$}
\DisplayProof &
\end{tabular}
\caption{Typing rules}
\label{fig:typingrules}
\end{figure}

%\blue{\textbf{Notation:} we denote as $\termconstructor{\Gamma}{\host}{\promote{-.f}}{\Ctype{\ptype}{I}{A}}$ the result of the promotion of a core term of the shape $\terconstrmix{\Gamma}{-}{\core}{f}{A}$.}
%
%The construction internal language of $\mV$ is standard for cartesian closed category, so for example, for every morphism $\mV(X,Y)$ gives a terms constructor in the internal language, with the usual equality rule.
Notice that in the dereliction rule (der) a host-term is ``derelicted'' to the core when applied to a (linear) core-term of the correct type.  This style is consistent with the standard way of presenting this kind of rule (e.g., see~\cite{RS20}). Similarly, the promotion rule (prom) ``promotes'' a term from the core to the host, by ``preserving'' the linear context. In the following, we denote by $\termconstructor{\Gamma}{\host}{\promote{-.f}}{\Ctype{\ptype}{I}{A}}$ the result of the promotion of a core term of the shape $\terconstrmix{\Gamma}{-}{\core}{f}{A}$. The duality between (der) and (prom) is formalized in Figure \ref{dualpromder}.

\subsection{Evaluation}\label{sec:evaluation}
In this section, we introduce an equational theory for $\VCcalcolo$.
First we present a substitution lemma. It is straightforward to prove (by induction on the derivations) that all substitutions  are admissible. 

Notice also that our substitutions are quite similar to the substitutions presented in the enriched effect calculus \texttt{EEC}, see \cite[Prop. 2.3]{EggerMS14}

{%First, we present three substitution rules. 
\begin{lemma}[Substitution]
The following hold:

\begin{itemize}

   \item (sub1) if $\termconstructor{\Gamma}{\host}{s}{X}$ and $\terconstrmix{\Gamma,x:X}{\Omega}{\core}{e}{A}$ then $\terconstrmix{\Gamma}{\Omega}{\core}{e[s/x]}{A}$;
    
    \item (sub2) if $\terconstrmix{\Gamma}{\Omega_1}{\core}{g}{A}$ and $\terconstrmix{\Gamma}{\Omega_2,a:A}{\core}{f}{B}$ then $\terconstrmix{\Gamma}{\Omega_1,\Omega_2}{\core}{f[g/a]}{B}$;
   
  \item (sub3) if $\termconstructor{\Gamma}{\host}{s}{X}$ and $\termconstructor{\Gamma,x:X}{\host}{t}{Y}$ then 
$\termconstructor{\Gamma}{\host}{t[s/x]}{Y}$.

\end{itemize}

\end{lemma}

Notice that $(sub1)$ allows us to substitute a host term for a variable in a core term.
%The rule $(sub2)$ models the substitution of a term (of the correct type) for a variable in the host language.  The rule $(sub3)$  models the substitution of a term (of the correct type) for a variable in the core language .   
In designing $(sub2)$, we follow the line adopted in languages like $\qwire$ and $\ewire$, which have some primitive notions of substitution for composing functions of $\ptype(\cdot,\cdot)$ type.
Finally, $(sub3)$ models the replacement of a term for a variable in the host language as well as in simply typed lambda calculus.

The $\beta$ and $\eta$ rules for the $\host$-terms are those of simply typed lambda calculus  with the connectives $\times$ and $\rightarrow$, see Appendix~\ref{app:evaluation rules}. %\blue{ Abbiamo i contesti moltiplicativi? Aggiungerei poi le chiusure contestuali che seguono (per evitare domande sulle strategie di riduzione), che dici?: $M\rightarrow N\Rightarrow ML\rightarrow NL$, $M\rightarrow N\Rightarrow LM\rightarrow LN$, $M\rightarrow N\Rightarrow \lambda M\rightarrow \lambda.N$. Forziamo la Congruenza?}

%\red{REBUTTAL Notice that the operational semantics that naturally arises by orienting the equational theory will be reasonably quite different from the QWIRE one. The crucial point is the different evaluation style we choose. In HC evaluation is allowed both in the H and in the C. In QEWire the evaluation is completely brought at the host level.}

The $\beta$ and $\eta$ rules for the $\core$-terms are those of the modality-free fragment of Linear Logic (without exponents) and are in Figure~\ref{subrules}. 

\begin{figure}[H]
\centering
\renewcommand\arraystretch{3}
%\small
\begin{tabular}{c}
%\AxiomC{$\terconstrmix{\Gamma}{\Omega_1,a:A}{\mC}{f}{B}$}
%\AxiomC{$\terconstrmix{\Gamma}{\Omega_2}{\mC}{g}{A}$}
%\BinaryInfC{$\terconstrmix{\Gamma}{\Omega_1,\Omega_2}{\mC}{(\lambda a:A.f)(g)=f[g/a]}{B}$}
%\DisplayProof \\
(el1) \AxiomC{$\terconstrmix{\Gamma}{\Omega_1,a:A,b:B}{\core}{f}{C}$}
\AxiomC{$\terconstrmix{\Gamma}{\Omega_2}{\core}{g}{A}$}
\AxiomC{$\terconstrmix{\Gamma}{\Omega_3}{\core}{h}{B}$}
\TrinaryInfC{$\terconstrmix{\Gamma}{\Omega_1,\Omega_2,\Omega_3}{\core}{\letbein{g\otimes h}{a\otimes b}{f}=f[g/a,h/b]}{C}$}
\DisplayProof \\
(el2)\AxiomC{$\terconstrmix{\Gamma}{\Omega_1,c:A\otimes B}{\core}{f}{C}$}
\AxiomC{$\terconstrmix{\Gamma}{\Omega_2}{\core}{g}{A\otimes B}$}
\BinaryInfC{$\terconstrmix{\Gamma}{\Omega_1,\Omega_2}{\core}{\letbein{g}{a\otimes b}{f[a\otimes b/c]}=f[g/c]}{C}$}
\DisplayProof \\
(el3)\AxiomC{$\terconstrmix{\Gamma}{\Omega}{\core}{f}{A}$}
\UnaryInfC{$\terconstrmix{\Gamma}{\Omega}{\core}{\letbein{\bullet}{\bullet}{f}=f}{A}$}
\DisplayProof \\
(el4)\AxiomC{$\terconstrmix{\Gamma}{\Omega_1,a:I}{\mC}{f}{A}$}
\AxiomC{$\terconstrmix{\Gamma}{\Omega_2}{\mC}{g}{I}$}
\BinaryInfC{$\terconstrmix{\Gamma}{\Omega_1,\Omega_2}{\mC}{\letbein{g}{\bullet}{f[\bullet /a]}=f[g/a]}{A}$}
\DisplayProof
\end{tabular}
\caption{Let-Evaluation rules for the core language $\core$}\label{subrules}
\end{figure}

We also add the following rules in Figure~\ref{dualpromder}, modelling the duality between promotion and dereliction operations, where $\Omega:=[a_1:A_1,\dots, a_n:A_n]$:

\begin{figure}[H]
\centering
\renewcommand\arraystretch{3}
%\small

%%%%NUOVE REGOLE    
\begin{tabular}{c}
%\textsf{(A)}\AxiomC{$\terconstrmix{\Gamma}{\Omega}{\core}{f}{A}$}
%\UnaryInfC{$\terconstrmix{\Gamma}{{\Omega}}{\core}{\derelict{\promote{a_1,\ldots,a_n.f},a_1\otimes\ldots\otimes a_n}=f}{A}$}
%\DisplayProof\\
(der-prom)\AxiomC{$\terconstrmix{\Gamma}{\Omega}{\core}{f}{A}$} \AxiomC{$\terconstrmix{\Gamma}{\Omega_1}{\core}{g_1}{A_1}\;\ldots\;\terconstrmix{\Gamma}{\Omega_n}{\core}{g_n}{A_n}$}
\BinaryInfC{$\terconstrmix{\Gamma}{{\Omega}}{\core}{\derelict{\promote{a_1,\ldots,a_n.f},g_1\otimes\ldots\otimes g_n}=f[g_1/a_1\ldots g_n/a_n]}{A}$}
\DisplayProof\\
(prom-der)\AxiomC{$\termconstructor{\Gamma}{\host}{f}{\Ctype{\ptype}{A}{B}}$}
\UnaryInfC{$\termconstructor{\Gamma}{\host}{\promote{a.\derelict{f,a}}=f}{\Ctype{\ptype}{A}{B}}$}
\DisplayProof
\end{tabular}
\caption{Promote and Derelict duality rules}\label{dualpromder}
\end{figure}

With respect to Benton's formulation  of the \LNL\ system, here we add the $\eta$ rules to the calculus following the presentation of~\cite{RCSILL}. Observe that this is necessary to get a completeness result.

Moreover, we highlight again that a crucial difference between the \LNL\ original presentation and our calculus $\VCcalcolo$ is that the interaction between the $\host$-language and $\core$ is \emph{not} symmetric, since we are giving a privileged position to the host language $\host$.  %From the semantic perspective, this asymmetry in the presentation of the syntax is reflected by the fact that we will not have an adjunction but just an instance of an enrichment of a category (see Section~\ref{sec:semantics}). 

Observe that if we consider the \emph{pure} $\core$-judgements, which are those of the form
\[ \terconstrmix{-}{\Omega}{\core}{f}{A}\]
 we obtain exactly a pure linear language, with tensor products.

Since we are programmatically defining  $\VCcalcolo$ as a minimal system, having a linear core is a natural starting choice.  This could appear very limiting and difficult to overcome. However, it is not the case, as shown in the following example.
% {We provide now an example that shows how to extend $\core$ to a non linear core.}
\begin{exa}[Breaking the linearity of the core $\core$, Part I]\label{ex:linpart1}
To extend $\core$ to a non-linear core, we have to force the tensor product to be a cartesian product. This can be done by adding to the core type system the  rules $(uv)$, $(pv)$, $(\pi1)$, $(\pi 2)$ (with mixed contexts) and (eventually) {new base types}.  Notice that old rules for tensor product are  derivable from the new ones. 
\end{exa}
{Before we proceed  to the presentation of $\VCcalcolo$, we recall the ``roadmap'' that motivates our investigation. Linearity is an important property when designing languages for quantum computing. At the same time, when designing multi-language frameworks, some of the languages involved might have linear features. Both of these lines of enquiry require us to know how to deal with a minimal linear system suitably embedded into a host language. We  designed $\VCcalcolo$ following these two intuitions. On the one hand, we follow the categorical logic tradition started by Benton. On the other hand, we also want to model some specific languages for quantum computing and in general for circuit manipulation.  At the same time, we take inspiration from the practice of embedded programming.  Our goal is not to provide a full description of host-core programming theory, but  to show how a principled minimal system such as  $\VCcalcolo$ works and can be extended to provide a useful basis for future systems.}

%\red{Qui valeria consigliava Next we describe our running example built upon the concrete syntax $\VCcalcolo$ from Example~\ref{ex:running1} ?Fatto, ma non abbiamo piu' l'istanza come esempio, quindi ho semplificato}
 
Next, we develop our running example built upon the concrete syntax of $\VCcalcolo$ (see Examples \ref{ex:running2}, \ref{ex:running3} and \ref{ex:running4}). %In the following we describe our running example, in order to clarify the possible interaction between  $\host$ and $\core$. %\blue{Non abbiamo piu' HC calligrafico come istanza}
%In the following example (and along this section) we carry on our running example, built upon the concrete syntax $\vcinstance$ from Example~\ref{ex:running1}. 
%We define a minimal, but yet interesting \emph{instance} of $\vcinstance$,  
We sketch how to equip $\VCcalcolo$ with constants and functions to support basic circuit definition and manipulation, suggesting a \emph{specialization} of $\VCcalcolo$ as a system hosting a (toy) hardware definition language.
{
\begin{remark}
We would like to emphasize that the complete definition of concrete instances of syntax for the definition of paradigmatic languages, as well as the study of operational semantics, are beyond the scope of this work. The exploration of these intriguing developments is left to future research. %The examples are intended to serve as support for the reader.
The examples are intended to be understood as support for the reader and as a tool to suggest possible connections with the literature on languages for circuit manipulation, which are particularly relevant for both classical and probabilistic, reversible, and quantum computations.
\end{remark}
}
\newcommand{\zero}{\mathbf{0}}
\newcommand{\uno}{\mathbf{1}}

\begin{exa}[A circuit core language ]\label{ex:running2}
%  In the minimal case the host can control the core processes inducing a direct form of interaction, but more complex bidirectional interaction can be modeled.

%\subsection{A core with circuits}
%Consider the core language $\icore$ from Example~\ref{ex:running1}. 

Let $\core^{*}$ be an extension of $\core$ for basic circuits definition and manipulation, and let $\host^{*}$ be the corresponding extension of $\host$.  %$\icore^{*}$ takes inspiration from 
The new terms and types of $\core^{*}$  are defined as follows: we denote by $\boldbit$ a new type representing bits, and by $\zero$ and  ${\uno}$ two constants typed as follows:

 $$
\begin{array}{cc}
\terconstrmix{\Gamma}{-}{\core}{\mathbf{\zero}}{\boldbit} & \qquad \terconstrmix{\Gamma}{-}{\core}{\mathbf{\uno}}{\boldbit} \\
 \end{array}
$$
 %  $\boldbit$  $\mathsf{not}, $\mathsf{and}$ and $\mathsf{cnot}$ (the classical controlled-not gate) 
Moreover, we consider some new functions $\mathsf{not}$, $\mathsf{and}$ and $\mathsf{cnot}$ (the controlled-not), representing  boolean gates, together with their typing judgements:

%\blue{
%$$
%\begin{array}[t]{l}
%\terconstrmix{\Gamma}{\Omega,a: \boldbit}{\core}{\mathsf{not}\, a}{\boldbit}\\  \terconstrmix{\Gamma}{\Omega,a: \boldbit, b:\boldbit}{\core}{\mathsf{and}\, (a,b)}{\boldbit} \\
%{     \terconstrmix{\Gamma}{\Omega, a: \boldbit, b:\boldbit}{\core}{\mathsf{cnot}\, (a,b)}{\boldbit\otimes\boldbit}  }
%\end{array}
%$$
%}

\begin{center}
    \renewcommand\arraystretch{3}
\begin{tabular}{c}
\AxiomC{$\terconstrmix{\Gamma}{\Omega}{\core}{f}{\boldbit}$}
\UnaryInfC{$\terconstrmix{\Gamma}{\Omega}{\core}{\mathsf{not}(f)}{\boldbit}$}
\DisplayProof
\\

\AxiomC{$\terconstrmix{\Gamma}{\Omega_1}{\core}{{f}}{\boldbit}$}
\AxiomC{$\terconstrmix{\Gamma}{\Omega_2}{\core}{{g}}{\boldbit}$}
\BinaryInfC{$\terconstrmix{\Gamma}{\Omega_1,\Omega_2}{\core}{\mathsf{and}(f,g)}{\boldbit}$}
\DisplayProof
\\

\AxiomC{$\terconstrmix{\Gamma}{\Omega_1}{\core}{{f}}{\boldbit}$}
\AxiomC{$\terconstrmix{\Gamma}{\Omega_2}{\core}{{g}}{\boldbit}$}
\BinaryInfC{$\terconstrmix{\Gamma}{\Omega_1,\Omega_2}{\core}{\mathsf{cnot}(f,g)}{\boldbit\otimes\boldbit}$}
\DisplayProof
\\

\end{tabular}
\end{center}

Combining these new operators with the substitution rule $(sub2)$ one can easily build, for example, the ``not-and'' $\mathsf{nand}$ operator as follows:

\begin{center}
  \renewcommand\arraystretch{3}
\AxiomC{$\terconstrmix{\Gamma}{\Omega_1}{\core}{{f}}{\boldbit}$}
\AxiomC{$\terconstrmix{\Gamma}{\Omega_2}{\core}{{g}}{\boldbit}$}
\BinaryInfC{$\terconstrmix{\Gamma}{\Omega_1,\Omega_2}{\core}{\mathsf{and}(f,g)}{\boldbit}$}
\UnaryInfC{$\terconstrmix{\Gamma}{\Omega_1,\Omega_2}{\core}{\mathsf{not}(\mathsf{and}(f,g))}{\boldbit}$}
\DisplayProof
\end{center}

Similarly, one can derive, for example, $\terconstrmix{\Gamma}{\Omega_1}{\core}{ \mathsf{and}((\mathsf{not}a_1),(\mathsf{and}\, a_2))  }{\boldbit}$:

\begin{center}
  \renewcommand\arraystretch{3}
\AxiomC{$\terconstrmix{\Gamma}{\Omega_1}{\core}{{h}}{\boldbit}$}
\UnaryInfC{$\terconstrmix{\Gamma}{\Omega_1}{\core}{\mathsf{not}(h)}{\boldbit}$}
\AxiomC{$\terconstrmix{\Gamma}{\Omega_2}{\core}{{f}}{\boldbit}$}
\AxiomC{$\terconstrmix{\Gamma}{\Omega_3}{\core}{{g}}{\boldbit}$}
\BinaryInfC{$\terconstrmix{\Gamma}{\Omega_2,\Omega_3}{\core}{\mathsf{and}(f,g)}{\boldbit}$}
\BinaryInfC{$\terconstrmix{\Gamma}{\Omega_1,\Omega_2,\Omega_3}{\core}{\mathsf{and}(\mathsf{not}(h),\mathsf{and}(f,g))}{\boldbit}$}
%\UnaryInfC{$\terconstrmix{\Gamma}{\Omega_1,\Omega_2}{\core}{\mathsf{not}(\mathsf{and}(f,g))}{\boldbit}$}
\DisplayProof
\end{center}

%\blue{ho commentato la derivazione del let, mi sembrava non aggiungesse tanto e istigasse richieste di altre codifiche, ma se pensi sia utile scommentala}
\commento
{From this latter, by applying $(tc)$ and $(let1c)$ we conclude:
\vspace{-5mm}\begin{scprooftree}{.9}
  \AxiomC{$\terconstrmix{\Gamma}{a'_2:\boldbit}{\core}{a'_2}{\boldbit}$}
\AxiomC{$\terconstrmix{\Gamma}{a''_2:\boldbit}{\core}{a''_2}{\boldbit}$}
\BinaryInfC{$\terconstrmix{\Gamma}{a'_2:\boldbit,a''_2:\boldbit}{\core}{a'_2\otimes a''_2}{\boldbit\otimes\boldbit }$}
\AxiomC{ $\terconstrmix{\Gamma}{a_1:\boldbit,a_2: \boldbit\otimes\boldbit}{\core}{ \mathsf{and}((\mathsf{not}a_1) \otimes(\mathsf{and}\, a_2))  }{\boldbit}$}
\BinaryInfC{$\terconstrmix{\Gamma}{a_1: \boldbit,a'_2: \boldbit,a''_2: \boldbit}{\core}{\letbein{a'_2\otimes a''_2}{a_2}{\mathsf{and}((\mathsf{not}a_1) \otimes(\mathsf{and}\, a_2)) }}{\boldbit}$}
\end{scprooftree}
}

Notice that the linear constraints on $\boldbit$ forbid the duplication and discarding of the variables of type $\boldbit$. This can be crucial if we want strictly control linearity in {contexts like} quantum circuits manipulation, where data duplication violates crucial physical properties such as the no-cloning property.

\end{exa}

In the following remarks, we discuss  how $\VCcalcolo$ can manage the  composition of core functions once a function has been promoted in the host part of the syntax. %From a language perspective, we can interpret % \blue{We comment on this design choice also in Section~\ref{sec:discussions}, focusing on the programming perspective. PENSIAMO SE FAR RIVIVERE IL DISCORSO MACRO DI QUESTA VERSIONE}
%\begin{rem}

 %\end{rem}

\begin{rem}[Composition]\label{remark formal composition}
\commento{\red{Given two $\host$-judgements of the form
\[\termconstructor{\Gamma}{\host}{f}{\Ctype{\ptype}{A}{B}} \mbox{ and }\termconstructor{\Gamma}{\host}{g}{\Ctype{\ptype}{B}{C}} \]
we can construct a term 
\[\termconstructor{\Gamma}{\host}{\enrichedcomp{f}{g}}{\Ctype{\ptype}{A}{C}} \]
which represents the composition in $\host$, given by the derivation}
}

%\blue{La comp ora \`e esplicitamente parametrica anche rispetto agli argomenti, diciamo f,g, che passiamo a s e t, sarebbe $\enrichedcomp{s,f}{t,g}$. }

Given the following terms
\[\termconstructor{\Gamma}{\host}{s}{\Ctype{\ptype}{A}{B}} \mbox{ and }\termconstructor{\Gamma}{\host}{t}{\Ctype{\ptype}{B}{C}} \mbox{ and } \terconstrmix{\Gamma}{a:A}{\core}{a}{A}\]
we can construct a new host-term  denoted by
\[\termconstructor{\Gamma}{\host}{\enrichedcomp{a.s}{t}}{\Ctype{\ptype}{A}{C}} \]
which represents the composition in $\host$, given by the following derivation:

\begin{prooftree}
\AxiomC{$\termconstructor{\Gamma}{\host}{s}{\Ctype{\ptype}{A}{B}}$}
\AxiomC{$\terconstrmix{\Gamma}{a:A}{\core}{a}{A}$}
\BinaryInfC{$\terconstrmix{\Gamma}{a:A}{\core}{\derelict{s,a}}{B}$}
\AxiomC{$\termconstructor{\Gamma}{\host}{t}{\Ctype{\ptype}{B}{C}}$}
%\AxiomC{$\terconstrmix{\Gamma}{b:B}{\core}{b}{B}$}
%\BinaryInfC{$\terconstrmix{\Gamma}{b:B}{\core}{\derelict{t,b}}{C}$}
\BinaryInfC{$\terconstrmix{\Gamma}{a:A}{\core}{\derelict{t,\derelict{s,a}}}{C}$}
\UnaryInfC{$\termconstructor{\Gamma}{\host}{\promote{a.\derelict{t,\derelict{s,a}}}}{\Ctype{\ptype}{A}{C}}$.}
\end{prooftree}

\end{rem}

\begin{rem}[Associativity of the composition]\label{remark associativity of composition}

Combining the axioms with the substitutions rules, we can derive the following equalities describing the composition rule

\begin{table}[H]
\centering
\renewcommand\arraystretch{3}
\begin{tabular}{c}
\AxiomC{$\termconstructor{\Gamma}{\host}{s}{\Ctype{\ptype}{A}{B}}$}
\AxiomC{$\termconstructor{\Gamma}{\host}{t}{\Ctype{\ptype}{B}{C}}$}
\AxiomC{$\termconstructor{\Gamma}{\host}{u}{\Ctype{\ptype}{C}{D}}$}
\AxiomC{$\terconstrmix{\Gamma}{a:A}{\core}{a}{A}$}
\QuaternaryInfC{$\termconstructor{\Gamma}{\host}{\enrichedcomp{a.\enrichedcomp{a.s}{t}}{u}=\enrichedcomp{a.s}{\enrichedcomp{b.t}{u}}}{\Ctype{\ptype}{A}{D}}$}
\DisplayProof \\
\AxiomC{$\termconstructor{\Gamma}{\host}{s}{\Ctype{\ptype}{A}{B}}$}
\AxiomC{$\termconstructor{\Gamma}{\host}{\id_A}{\Ctype{\ptype}{A}{A}}$}
\AxiomC{$\terconstrmix{\Gamma}{a:A}{\core}{a}{A}$}
\TrinaryInfC{$\termconstructor{\Gamma}{\host}{\enrichedcomp{a.\id_A}{s}=s}{\Ctype{\ptype}{A}{B}}$}
\DisplayProof \\
\AxiomC{$\termconstructor{\Gamma}{\host}{s}{\Ctype{\ptype}{A}{B}}$}
\AxiomC{$\termconstructor{\Gamma}{\host}{\id_B}{\Ctype{\ptype}{B}{B}}$}
\AxiomC{$\terconstrmix{\Gamma}{a:A}{\core}{a}{A}$}
\TrinaryInfC{$\termconstructor{\Gamma}{\host}{\enrichedcomp{a.s}{\id_B
}=s}{\Ctype{\ptype}{A}{B}}$}
\DisplayProof
\end{tabular}
\end{table}
\noindent
where the $\host$-terms $\termconstructor{\Gamma}{\host}{\id_A}{\Ctype{\ptype}{A}{A}}$ and  $\termconstructor{\Gamma}{\host}{\id_B}{\Ctype{\ptype}{B}{B}}$ denote the promotion of $\terconstrmix{\Gamma}{a:A}{\core}{a}{A}$ and $\terconstrmix{\Gamma}{b:B}{\core}{b}{B}$ respectively, i.e. $\id_A:=\promote{a.a}$ and $\id_B:=\promote{b.b}$.

%\commento{Qui spiegherei ancora un po' il comp, bisognerebbe introdurre Id e forse suggerire che logicamente cattura il CUT}

We start by checking the associativity law: by definition, the term $\termconstructor{\Gamma}{\host}{\enrichedcomp{a.\enrichedcomp{a.s}{t}}{u}}{\Ctype{\ptype}{A}{D}}$ is given by
$$\termconstructor{\Gamma}{\host}{\promote{a.\derelict{u,\derelict{\promote{a.\derelict{t,\derelict{s,a}}},a}}}}{\Ctype{\ptype}{A}{\blue{D}}}$$
and, applying the duality rule of promotion and dereliction \textsf{(prom-der)}, we have that this is equal to
$$\termconstructor{\Gamma}{\host}{\promote{a.\derelict{u,\derelict{t,\derelict{s,a}}}}}{\Ctype{\ptype}{A}{D}}.$$
Similarly, we have that the term $\termconstructor{\Gamma}{\host}{\enrichedcomp{a.s}{\enrichedcomp{b.t}{u}}}{\Ctype{\ptype}{A}{D}}$ is given by
$$\termconstructor{\Gamma}{\host}{\promote{a.\derelict{\promote{b.\derelict{u,\derelict{t,b}}},\derelict{s,a}}}}
{\Ctype{\ptype}{A}{D}}$$

{Now, since $\derelict{s,a}$ is of type $B$, we can soundly apply the duality rule \textsf{(der-prom)} obtaining $$\termconstructor{\Gamma}{\host}{\promote{a.\derelict{u,\derelict{t,\derelict{s,a}}}}}{\Ctype{\ptype}{A}{D}}.$$
 }
\commento{\red{Now notice that, since $\promote{b.\derelict{u,\derelict{t,b}}}$ is a host-term and hence the core-variable $b$ does not freely occur in such a term,  this term can be written as
$$\termconstructor{\Gamma}{\host}{\promote{a.\derelict{\promote{b.\derelict{u,\derelict{t,b}}},b}[\derelict{s,a}/b]}}
{\Ctype{\ptype}{A}{C}}$$
therefore we can apply the duality rule of promotion and dereliction,  obtaining
$$\termconstructor{\Gamma}{\host}{\promote{a.\derelict{u,\derelict{t,b}}[\derelict{s,a}/b]}}
{\Ctype{\ptype}{A}{C}}$$
and this is exactly
$$\termconstructor{\Gamma}{\host}{\promote{a.\derelict{u,\derelict{t,\derelict{s,a}}}}}{\Ctype{\ptype}{A}{C}}$$
}}
Therefore, this shows that the composition is associative. 

Similarly, we show that $\termconstructor{\Gamma}{\host}{\enrichedcomp{a.s}{\id_B
}=s}{\Ctype{\ptype}{A}{B}}$ can be derived. By definition, the term $\termconstructor{\Gamma}{\host}{\enrichedcomp{a.s}{\id_B}}{\Ctype{\ptype}{A}{B}}$ is
$$\termconstructor{\Gamma}{\host}{\promote{a.\derelict{\promote{b.b},\derelict{s,a}}}}{\Ctype{\ptype}{A}{B}}$$

{Notice that $\derelict{s,a}$ has type $B$, so we can apply the duality rule \textsf{der-prom}, obtaining $$\termconstructor{\Gamma}{\host}{\promote{a.\derelict{s,a}}}{\Ctype{\ptype}{A}{B}}$$ Finally, we can conclude that this is exactly $$\termconstructor{\Gamma}{\host}{s}{\Ctype{\ptype}{A}{B}}$$ by means of the duality rule $\textbf{(prom-der)}$.
}

\commento{
\red{and this can be written as
$$\termconstructor{\Gamma}{\host}{\promote{a.\derelict{\promote{b.b},b}[\derelict{s,a}/b]}}{\Ctype{\ptype}{A}{B}}$$
Therefore we can apply the duality rules of promotion and dereliction, obtaining
$$\termconstructor{\Gamma}{\host}{\promote{a.\derelict{s,a}}}{\Ctype{\ptype}{A}{B}}$$
and, applying the duality rule again, we can conclude that this is precisely 
$$\termconstructor{\Gamma}{\host}{s}{\Ctype{\ptype}{A}{B}}$$
}
}
{Finally, the equality $\termconstructor{\Gamma}{\host}{\enrichedcomp{a.\id_A}{s
}=s}{\Ctype{\ptype}{A}{B}}$ can be derived  as follows. Consider the explicit definition of ${\enrichedcomp{a.\id_A}{s
}}$:
$$\termconstructor{\Gamma}{\host}{\promote{a.\derelict{s,\derelict{\promote{a.a},a}}}}{\Ctype{\ptype}{A}{B}}$$
By appliyng the \textsf{(der-prom)} rule to the (core) subterm $\derelict{\promote{a.a},a}$ we obtain 
$$\termconstructor{\Gamma}{\host}{\promote{a.\derelict{s,{a}}}}{\Ctype{\ptype}{A}{B}}$$
and then by means of the \textsf{(prom-der)} rule we can conclude 
$$\termconstructor{\Gamma}{\host}{s}{\Ctype{\ptype}{A}{B}}$$.}
\commento{
\red{Finally, the equality $\termconstructor{\Gamma}{\host}{\enrichedcomp{a.\id_A}{s
}=s}{\Ctype{\ptype}{A}{B}}$ can be derived by applying directly the duality rule of promotion and dereliction to the explicit definition of the composition.}}
\end{rem}

\noindent
In Examples~\ref{ex:running3} and ~\ref{ex:running4} we suggest how to accommodate the syntax of the host language for
the extension defined in Example~\ref{ex:running2}. In particular, we focus on the circuits manipulation at host level.

\newcommand{\senrichedcomp}[2]{\mathbf{seq}(#1,#2)}

\begin{exa}[Hosted Circuits]\label{ex:running3}

%\blue{Nota: in tutto il paper controllare di usare l'underscore per il promote da contesto vuoto (costanti).}

Let us consider the extension of $\VCcalcolo$  defined in Example \ref{ex:running2}. Employing the promotion rule $(prom)$, we can represent the new core constants $\zero$ and $\uno$ at the host level:

$$
\begin{array}{cc}
  \termconstructor{\Gamma}{\host}{\promote{\mathbf{\mbox{\_\;}.\zero}}}{\Ctype{\ptype}{I}{\boldbit}}  &  \termconstructor{\Gamma}{\host}{\promote{\mathbf{\mbox{\_}.\uno}}}{\Ctype{\ptype}{I}{\boldbit}}\\
 %               \termconstructor{\Gamma}{\host}{\promote{\mathsf{not} \, a}}{\Ctype{\ptype}{\boldbit}{\boldbit}}\\
%  \termconstructor{\Gamma}{\host}{\promote{\mathbf{\uno}}}{\Ctype{\ptype}{\tenscont{\Omega}}{\boldbit}}  &
%   \termconstructor{\Gamma}{\host}{\promote{\mathsf{and}\, a}}{\Ctype{\ptype}{\boldbit\otimes\boldbit}{\boldbit}}     \\
% \termconstructor{\Gamma}{\host}{\promote{\mathsf{cnot}\, a}} {\Ctype{\ptype}{\boldbit\otimes\boldbit}{\boldbit}}  &
 %    \termconstructor{\Gamma}{\host}{\mathsf{not}(\mathsf{and}\, a)}{\Ctype{\ptype}{\boldbit\otimes\boldbit}{\boldbit} }\\
%\multicolumn{2}{c}{     \termconstructor{\Gamma}{\host}{\promote{ \letbein{a'_2\otimes a''_2}{a_2}{\mathsf{and}((\mathsf{not}a_1) \otimes(\mathsf{and}\, a_2)) }}  }{\Ctype{\ptype}{\boldbit\otimes\boldbit\otimes\boldbit}{\boldbit} }}
\end{array}
$$
Similarly, one can formally reason at the host-level about circuits and gates (see Example~\ref{ex:running2}), by promoting the new core terms,  e.g,
%given the core judgements $\terconstrmix{\Gamma}{\Omega_1}{\core}{\mathsf{not}(f)}\boldbit$ and $\terconstrmix{\Gamma}{\Omega_2,\Omega_3}{\core}{\mathsf{and}(g,h)}\boldbit$, where $\Omega_1=\{a_1,\ldots,a_n\}$, $\Omega_2=\{b_1,\ldots,b_m\}$, $\Omega_3=\{c_1,\ldots,c_k\}$, one has:}

 %$\termconstructor{\Gamma}{\host}{\promote{a.\mathsf{not} \, a}}{\Ctype{\ptype}{\boldbit}{\boldbit}}$ and   $ \termconstructor{\Gamma}{\host}{\promote{a.\mathsf{cnot}\, a}} {\Ctype{\ptype}{\boldbit\otimes\boldbit}{\boldbit}}$.\\
$$\termconstructor{\Gamma}{\host}{\promote{a_1\ldots a_n.\mathsf{not} \,( f)}}{\Ctype{\ptype}{\otimes_\Omega}{\boldbit}}$$ and   $$ \termconstructor{\Gamma}{\host}{\promote{b_1\ldots b_m,c_1\ldots c_k.\mathsf{cnot}\,(g,h)}} {\Ctype{\ptype}{\otimes_{\Omega_1,\Omega_2}}{\boldbit}}$$
where $\Omega=[a_1:A_1\ldots a_n:A_n] $, $\Omega_1=[b_1:B_1\ldots b_m:B_m] $ and  $\Omega_3=[c_1\ldots c_k] $ are the linear contexts of $f$, $g$, and $h$ respectively and we denote $A_1\otimes \ldots \otimes A_n\otimes B_1 \otimes\ldots  \otimes B_m$ as $\otimes_{\Omega_1,\Omega_2}$.

\commento{
\begin{remark}
\blue{To evaluate a circuit promoted to the host, we can directly use  the $\derelict{}$ rule, and  apply it to a suitable core input. Notice also that a promoted constant $\zero$ or $\uno$ has to be derelicted before being passed as an argument of a (suitable) host term, for example a promoted core-term representing a circuit. Promoted constants plays a role when we move toward a more complex interaction between host and core and, in general, if we want to represent in the host both core circuits and inputs. We show an example in Example \ref{ex:running3}.}
\end{remark}
}

The language $\host^{*}$ can also manage a flow-control on (core) circuits:
 for example, it can erase and duplicate circuits and it can compose circuits in parallel and sequence.

\newcommand{\circeval}[2]{\mathsf{eval}(#1,#2)}

Given the host terms $\termconstructor{\Gamma}{\host}{f}{\Ctype{\ptype}{A}{B}}$,  $\termconstructor{\Gamma}{\host}{g}{\Ctype{\ptype}{C}{D}}$,  $\termconstructor{\Gamma}{\host}{h}{\Ctype{\ptype}{B}{D}}$, and the core terms {$\terconstrmix{\Gamma}{a:A}{\core}{a}{A}$},  $\terconstrmix{\Gamma}{c:C}{\core}{c}{C}$, we can infer derivations
$$\enrichedparallel{a.f}{c.g}::=\promote{a,c.(\derelict{f,a}\otimes\derelict{g,c})}$$
$$\enrichedcomp{a.f}{h}::=\promote{a.(\derelict{f,\derelict{f,a}})}$$
which encode parallelization and sequentialization of core circuits respectively (for suitable given inputs). The term $\enrichedparallel{a.f}{c.g}$ has type $\Ctype{\ptype}{A\otimes C}{B\otimes D}$ and the term $\enrichedcomp{a.f}{h}$ has type $\Ctype{\ptype}{A}{ D}$ .
Notice that we encode circuits sequentialization exactly as the  $\enrichedcomp{-.f}{g}$ in Remark~\ref{remark formal composition}.

In Figure \ref{tab:circComp} we show the derivation for the closed host term
$${\lambda x_0: \Ctype{\ptype}{A_0}{B_0}.\lambda x_1: \Ctype{\ptype}{A_1}{B_1}.\enrichedparallel{a_0.x_0}{a_1.x_1}}$$ with  core inputs $a_0:A_0$ and $a_1:A_1$. 
%The term is parametric w.r.t. the circuits one aim to put in parallel and evaluate. 
%
This term has  type 
$\arrowtype{\Ctype{\ptype}{A_0}{B_0}}{\arrowtype{\Ctype{\ptype}{A_1}{B_1}}{\Ctype{\ptype}{A_0\otimes A_1}{B_0\otimes B_1 }}}$.
 Notice that the core variables $a_0$ and $a_1$ do not freely occur in the host. %so they are morally bounded.
%$\arrowtype{\Ctype{\ptype}{A_0}{B_0}}{\arrowtype{\Ctype{\ptype}{A_1}{B_1}}{\Ctype{\ptype}{A_0\otimes A_1}{B_0\otimes B_1 }}}$.
For the sake of space, 
we write $\Gamma_{\!01}$ to denote $x_0: \Ctype{\ptype}{A_0}{B_0},x_1: \Ctype{\ptype}{A_1}{B_1}$.

\begin{figure}[H]%%%%%%%%%%%%%%%%%%%%%%%%%%%%%%%%%%%%%%%%%%%%%%%%%%%%%%%%%%%%%%%%%%%%%%%%%%%%%%%%%%%%%%%%%%%%%%%%%%%%%%%%%%%%%% \centering
 \scalebox{.9}{ }\\[-1cm]
\begin{scprooftree}{.8}
\hspace*{-1cm}\AxiomC{$\termconstructor{\Gamma_{\!01}}{\host}{x_0}{\Ctype{\ptype}{A_0}{B_0}}$}
\AxiomC{$\terconstrmix{\Gamma_{\!01}}{a_0:A_0}{\core}{{a_0}}{A_0}$}
\BinaryInfC{$\terconstrmix{\Gamma_{\!01}}{a_0:A_0}{\core}{\derelict{x_0,a_0}}{B_0}$}
\AxiomC{$\termconstructor{\Gamma_{\!01}}{\host}{x_1}{\Ctype{\ptype}{A_1}{B_1}}$}
\AxiomC{$\terconstrmix{\Gamma_{\!01}}{a_1:A_1}{\core}{{a_1}}{A_1}$}
\BinaryInfC{$\terconstrmix{\Gamma_{\!01}}{a_1:A_1}{\core}{\derelict{x_1,a_1}}{B_1}$}
\BinaryInfC{$\terconstrmix{\Gamma_{\!01}}{a_0:A_0,a_1:A_1}{\core}{\derelict{x_0, a_0}\otimes \derelict{x_1, a_1}}{B_0\otimes B_1}$}
\UnaryInfC{$\termconstructor{\Gamma_{\!01}}{\host}{\enrichedparallel{a_0.x_0}{a_1.x_1}}{\Ctype{\ptype}{A_0\otimes A_1}{B_0\otimes B_1 }}$}
\UnaryInfC{$\termconstructor{x_0: \Ctype{\ptype}{A_0}{B_0}}{\host}{\lambda x_1: \Ctype{\ptype}{A_1}{B_1}.\enrichedparallel{a_0.x_0}{a_1.x_1}}
                {\arrowtype{\Ctype{\ptype}{A_1}{B_1}}{\Ctype{\ptype}{A_0\otimes A_1}{B_0\otimes B_1 }}}$}
\UnaryInfC{$\termconstructor{}{\host}{\lambda x_0: \Ctype{\ptype}{A_0}{B_0}.\lambda x_1: \Ctype{\ptype}{A_1}{B_1}.\enrichedparallel{a_0.x_0}{a_1.x_1}}
                {\arrowtype{\Ctype{\ptype}{A_0}{B_0}}{\arrowtype{\Ctype{\ptype}{A_1}{B_1}}{\Ctype{\ptype}{A_0\otimes A_1}{B_0\otimes B_1 }}}}$}
\end{scprooftree} %%%%%%%%%%%%%%%%%%%%%%%%%%%%
%\begin{prooftree} %%%%%%%%%%%%%%%%%%%%%%
%\AxiomC{$\termconstructor{x: \Ctype{\ptype}{A}{B}}{\host}{x}{\Ctype{\ptype}{A}{B}}$}
%\UnaryInfC{$\terconstrmix{x: \Ctype{\ptype}{A}{B}}{a:A}{\core}{\derelict{x}}{B}$}
%\AxiomC{$\termconstructor{z:\Ctype{\ptype}{B}{C}}{\host}{z}{\Ctype{\ptype}{B}{C}}$}
%\UnaryInfC{$\terconstrmix{z:\Ctype{\ptype}{B}{C}}{b:B}{\core}{\derelict{x}}{C}$}
%\BinaryInfC{$\terconstrmix{x: \Ctype{\ptype}{A}{B},z:\Ctype{\ptype}{B}{C} }{a:A}{\core}{\derelict{g}[\derelict{f}/b]}{C}$}
%\UnaryInfC{$\termconstructor{x: \Ctype{\ptype}{A}{B},z:\Ctype{\ptype}{B}{C} }{\host}{\enrichedcomp{x}{z}}{\Ctype{\ptype}{A}{C}}$}
%\UnaryInfC{$\termconstructor{x: \Ctype{\ptype}{A}{B} }{\host}{\lambda z:\Ctype{\ptype}{B}{C} .\enrichedcomp{x}{z}}{\arrowtype{\Ctype{\ptype}{B}{C} }{\Ctype{\ptype}{A}{C}}}$}
%\UnaryInfC{$\termconstructor{ }{\host}{\lambda x: \Ctype{\ptype}{A}{B}.\lambda z:\Ctype{\ptype}{B}{C} .\enrichedcomp{x}{z}}{
%                                                                             \arrowtype{\Ctype{\ptype}{A}{B}}{\arrowtype{\Ctype{\ptype}{B}{C} }{\Ctype{\ptype}{A}{C}}}}$}
%\end{prooftree}https://www.overleaf.com/project/63ce698118a10423bc0b413d
\caption{Host type derivation for parametric parallellelization of circuits.}
\label{tab:circComp}
\end{figure}%%%%%%%%%%%%%%%%%%%%%%%%%%%%%%%%%%%%%%%%%%%%%%%%%%%%%%%%%%%%%%%%%%%%%%%%%%%%%%%%%%%%%%%%%%%%%%%%%%%%%%%%%%%%%%
Similarly, we can also derive the closed term 
$${\lambda x: \Ctype{\ptype}{A}{B}.\lambda z:\Ctype{\ptype}{B}{C} .\enrichedcomp{a.x}{z}}$$
which has type ${\arrowtype{\Ctype{\ptype}{A}{B}}{\arrowtype{\Ctype{\ptype}{B}{C} }{\Ctype{\ptype}{A}{C}}}}$.
\end{exa}
\noindent
Example~\ref{ex:running3} shows how a simply typed language hosting a minimalistic circuit design calculus allows some interesting manipulation of core programs and, in some sense, a limited and oriented form of interoperability. 

%Clearly, the addition of syntactic sugar to the host or, more interestingly, the extension of the  expressive power of  both host and core languages provides more powerful and flexible host-core languages. This is out of the scope of this paper and left to future work (see Section\ref{sec:conclusions}). Nevertheless, we conclude the section with a concluding development of our running example. See Section~\ref{sec:conclusions} for further consideration about possible host language extensions and improvements of the communication mechanism between $\host^{*}$ and $\core^{*}$.

%We conclude the section with a development of our running example.

\begin{exa}[Controlling Circuits, some reflections]\label{ex:running4}
Consider the instance $\host^{*}$ from Example~\ref{ex:running3} and assume to add boolean constants $\mathsf{true}$ and $\mathsf{false}$, both typed $\boldbool$ (we omit the obvious typing judgements), and  also the $\texttt{if}\_ \texttt{then}\_\texttt{else}\_$ construct, typed by the rule
%\blue{\`e un'ovviet\`a... la lasciamo?}
\begin{prooftree}
  \AxiomC{$\termconstructor{\Gamma}{\host}{t}{\boldbool}$}
\AxiomC{$\termconstructor{\Gamma}{\host}{s_i}{X}$}
\BinaryInfC{$\termconstructor{\Gamma}{\host}{  \texttt{if} \,t\, \texttt{then} \,s_0\, \texttt{else} \,s_1\, }{X}$}
\end{prooftree}
Notice that we can use the $\texttt{if}\_ \texttt{then}\_\texttt{else}\_$ for the control of core programs. For example, given two circuits and according to the guard evaluation, we can decide  which circuit one should modify and use. Let $\Gamma$ be $x:\boldbool, f:\Ctype{\ptype}{\boldbit\otimes\boldbit}{\boldbit\otimes\boldbit}, g:\Ctype{\ptype}{\boldbit\otimes\boldbit}{\boldbit\otimes\boldbit}$. 

\commento{ Let $\Gamma$ be $x:\boldbool, f:\Ctype{\ptype}{\boldbit\otimes\boldbit}{\boldbit\otimes\boldbit}, g:\Ctype{\ptype}{\boldbit\otimes\boldbit}{\boldbit\otimes\boldbit}$. 
 }
It is easy to verify that the following judgments, which enable the representation of both input and core programs at the host level, are well-typed:% (as a title of example we show the derivation of the third one in Appendix \ref{app:evaluation rules}):

$$
\begin{array}[t]{l}
%  \termconstructor{\Gamma}{\core}{\derelict{{f,\mathbf{\zero}\otimes\mathbf{\zero}}}}{\boldbit\otimes\boldbit}\\
  % \termconstructor{\Gamma}{\host}{\promote{\mathbf{\zero}\otimes\mathbf{\zero}.\derelict{{f,\mathbf{\zero}\otimes\mathbf{\zero}}}}}{\Ctype{\ptype}{\boldbit\otimes\boldbit}{\boldbit\otimes\boldbit}} \mbox{identit\`a}\\
   \termconstructor{\Gamma}{\host}{\enrichedcomp{a.\promote{-.\zero\otimes\zero}}{f}}{\Ctype{\ptype}{I}{\boldbit\otimes\boldbit}} \\
%  \termconstructor{\Gamma}{\host}{\enrichedcomp{\mathbf{\uno}\otimes\mathbf{\uno}.f}{g} }{\boldbit\otimes\boldbit}\\
 { \termconstructor{\Gamma}{\host}{\enrichedcomp{a.\promote{-.\uno\otimes\uno}}{g}}{\Ctype{\ptype}{I}{\boldbit\otimes\boldbit}} }\\
  { \termconstructor{\Gamma}{\host}{\enrichedcomp{b.\promote{-.\uno\otimes\uno}}{\enrichedcomp{c.f}{g}}}{\Ctype{\ptype}{I}{\boldbit\otimes\boldbit}} }\\
 % \termconstructor{\Gamma}{\host}{\enrichedcomp{b.f}{g} }{\Ctype{\ptype}{\boldbit\otimes\boldbit}{\boldbit\otimes\boldbit}}
\end{array}
$$
%
%Let $B_0$ be the program $ \enrichedcomp{f_0}{}$
{ where $a: I$, $b: I$ and $c:\boldbit\otimes\boldbit$. The first term, call it  $s$,  represents the application of the program $f$ to the (promoted)  input $\promote{-.\zero\otimes\zero}$. The second, call it $t$, represents the application of the program $g$ to the (promoted)  input $\promote{-.\uno\otimes\uno}$. The third, call it $w$, models the application  to a promoted input $\promote{-.\uno\otimes\uno}$ of the sequentialization of the program (circuits) $f$ and $g$.}
 %Notice that the $\enrichedcomp{-.-}{-}$ (and the subtended $\derelict{-,-}$ rule), also requires a pure core term: this is the role of the core variable $a$. 
 %As observed in the previous example, one could directly derelict $s$ by applying it to a core term, e.g. ${\mathbf{\zero}\otimes\mathbf{\zero}}$ (obtaining something of the shape $\derelict{s,\mathbf{\zero}\otimes\mathbf{\zero}}$), but this yields a core judgement. 
% This allows the representation of both input and core programs within the host system.
%(NOTA:   core term, seconda riga lo rappresento a livello host perch\`e altrimenti non posso codificare un termine parametrico. Alternativa: $$ \termconstructor{\Gamma}{\host}{\enrichedcomp{a}{\promote{-.\zero\otimes\zero},f}}{\Ctype{\ptype}{I}{\boldbit\otimes\boldbit}})$$ 
%oppure
%$$ \termconstructor{\Gamma}{\host}{\enrichedcomp{\mathbf{\zero}\otimes\mathbf{\zero}.id_A}{{f}} }{\boldbit\otimes\boldbit})$$.
%The second term, call it $t$, applies to a core input the sequentialization of programs (circuits) $f$ and $g$. %For example, we can write   $\termconstructor{\Gamma}{\host}{\enrichedcomp{\mathbf{\uno}\otimes\mathbf{\uno}.f}{g} }{\boldbit\otimes\boldbit}$ as a particular case.

Now, let us endeavor to compose a conditional construct that assesses two distinct programs, effectively corresponding to two disparate circuit configurations

For example, we can  derive the term $\mathtt{IfCirc}$, defined as

{$$\lambda x.\lambda f.\lambda g.
\texttt{if}\,x\, \texttt{then}\;\enrichedcomp{a.\promote{-.\zero\otimes\zero}}{f}\;$$
$$
 \texttt{else}\;{\enrichedcomp{b.\promote{-.\uno\otimes\uno}}{\enrichedcomp{c.f}{g}}}$$}

\commento{
\red{
$$%\lambda x:\boldbit.\lambda f:\Ctype{\ptype}{\boldbit\otimes\boldbit}{\boldbit\otimes\boldbit}. \lambda g:\Ctype{\ptype}{\boldbit\otimes\boldbit}{\boldbit\otimes\boldbit}.
\lambda x.\lambda f.\lambda g.
\texttt{if}\,x\, \texttt{then}\;\enrichedcomp{a.\promote{-.\zero\otimes\zero}}{f}\; \texttt{else}\;{\enrichedcomp{b.f}{g}}$$}}
where, for the sake of readability, we omitted the type of lambda-abstracted variables. The term $\mathtt{IfCirc}$ performs a simple computation: evaluates a guard (the parameter passed to the bound variable $x$) and, according to the result, evaluates the subprogram $\lambda f. s$ or \blue{$\lambda g. w$}. 
 It is direct to check that this term (from the empty context) has the following type:
{
$$\arrowtype{\boldbit}{\arrowtype{ \Ctype{\ptype}{\boldbit\otimes\boldbit}{\boldbit\otimes\boldbit} }{
     \arrowtype{\Ctype{\ptype}{\boldbit\otimes\boldbit}{\boldbit\otimes\boldbit}}{\Ctype{\ptype}{I}{\boldbit\otimes\boldbit}}}}$$
    } \end{exa}

Example~\ref{ex:running4} shows a very basic form of control by the host on the core.  According to further extensions of the core $\core$, one can define more significant control operations.
For example, if we move to a Turing complete host, or at least to a typed system as expressive as G\"odel System $\mathcal{T}$, we can define interesting circuit construction operators such as the parametric circuit sequentialization and parallelization (which take a numeral $n$ as input and iterate the operation $n$ times). These operations are quite useful in circuit construction: for instance, in quantum programming languages, they allow us to encode parametric versions of well-known quantum algorithms~\cite{PPZ19}.
We postpone this intriguing topic to future endeavors (refer to Section \ref{sec:conclusions}).

\subsection{The category $\catvcalcolo$}
We conclude this section with a key definition that allows us to formulate our main result, the internal language Theorem \ref{theorem internal laguage}.  To prove the strong correspondence between syntax and models, we first define the category $\catvcalcolo$, whose objects are theories (Definition~\ref{def:theories}) and whose morphisms are translations (Definition~\ref{def:translation}). In this section, we strictly follow
the notation of~\cite{RCSILL}.%, we define the $\VCcalcolo$-theories and translations. 

\begin{defi}[$\VCcalcolo$-\bemph{theory}]\label{def:theories}
%Given a typed calculus $\VCcalcolo$, 
A typed system $\theory$ is a $\VCcalcolo$-\bemph{theory} if it is an extension of $\VCcalcolo$ with \bemph{proper}-$\theory$-\bemph{axioms}, namely with new base type symbols, new type equality rules, new term symbols and new equality rules of terms.
\end{defi}

\begin{defi}[$\VCcalcolo$-\bemph{translation}]\label{def:translation}
Given two $\VCcalcolo$-theories $\theory_1$ and $\theory_2$, a $\VCcalcolo$-\bemph{translation} $M$ is a function from types and terms of $\theory_1$ to types and terms of $\theory_2$ preserving type and term judgements, which means that it sends a type $X$ for which $\typeconstructor{\host}{X}$ is derivable in $\theory_1$ to a type $M(X)$ such that $\typeconstructor{\host}{M(X)}$ is derivable in $\theory_2$, it sends a type $A$ for which $\typeconstructor{\core}{A}$  is derivable in $\theory_1$ to a type $M(A)$ such that $\typeconstructor{\core}{M(A)}$ is derivable in $\theory_2$, it sends a $\host$-term  $t$ such that $\termconstructor{\Gamma}{\host}{t}{X}$ is derivable in $\theory_1$ to a $\host$-term $M(t)$ such that
\[\termconstructor{M(\Gamma)}{\host}{M(t)}{M(X)}\]
is derivable in $\theory_2$, where $M(\Gamma)\equiv [x_1:M(X_1),\dots,x_n:M(X_n)]$ if $\Gamma\equiv [x_1:X_1,\dots, x_n:X_n]$, and sending a $\core$-term $f$ for which $\terconstrmix{\Gamma}{\Delta}{\core}{f}{A}$ is derivable in $\theory_1$ to a typed term $M(f)$ such that 
\[\terconstrmix{M(\Gamma)}{M(\Delta)}{\core}{M(f)}{M(A)}\]
is derivable in $\theory_2$, where $M(\Delta)\equiv [a_1:M(A_1),\dots,a_m:M(A_m)]$ if $\Delta\equiv [a_1:A_1,\dots,a_m:A_m]$, satisfying in particular
\[M(x)=x \mbox{ and } M(a)=a\]
and preserving the $\VCcalcolo$-types and terms constructors, and their equalities.
\end{defi}

\begin{defi}[Category of $\VCcalcolo$-theories]\label{def category of theories}
The category $\catvcalcolo$ has $\VCcalcolo$-theories as objects and $\VCcalcolo$-translations as morphisms.
\end{defi}
%\begin{rem}
%
Observe that a translation between the two host parts of theories induces a \emph{change of base} for the core language, which means the core part of the first host language can be embedded in the second one using the translation between the hosts. In particular a translation can be formally split into two components: one acting on the host part, and the other one acting on the core language induced by the change of base. We will provide a formal definition of change of base in the next section where we introduce the $\VCcalcolo$-semantics.
%\end{rem}

\section{Categorical semantics}\label{sec:semantics}

In this section we show that a model for $\VCcalcolo$ is given by a pair $(\mV,\mC)$, where $\mV$ is a cartesian closed category, and $\mC$ is a $\mV$-enriched symmetric monoidal category, see \cite{BCECT}. 
%\red{\textbf{Notation}: in this work by  $\mV$-enriched symmetric monoidal category we mean a symmetric monoidal category enriched in $\mV$.}
Moreover we prove a correspondence between the typed calculus $\VCcalcolo$ and its categorical models via the notion of \emph{internal language}.

Recall that $\VCcalcolo$ provides  an internal language of the category   $\modelcat{\VCcalcolo}$ of its models if one proves an  equivalence between $\modelcat{\VCcalcolo}$ and the category $\catvcalcolo$ of the theories of the language. In the categorical logic literature several examples and applications of the notion of internal language can be found, e.g. see~\cite{IHOCL,CLP,SAE}. 
In this paper we mainly follow~\cite{RCSILL,milly2}, where the authors discuss models and morphisms for Intuitionistic Linear Logic (\texttt{ILL}), Dual Intuitionistic Linear Logic (\texttt{DILL}) and Linear-Non-Linear Logic (\LNL). 
The leading idea of ~\cite{RCSILL,milly2} is that soundness and completeness theorems are not generally sufficient to identify the most appropriate class of denotational models for a typed calculus, unless, as pointed out in~\cite{RCSILL}, the same typed calculus provides the internal language of the models we are considering, as anticipated in the introduction.

In Section \ref{sec:travail} we also use the internal language theorem to answer some questions about  host-core languages, showing how this notion is useful also out from its usual range of applications. 

Before we proceed to state and prove technical definitions and results, we sketch the key steps necessary to achieve the goal. 
Given the typed calculus $\VCcalcolo$, consider its category of theories $\catvcalcolo$ (as  in Definition \ref{def category of theories}). 
Our purpose is to define \bemph{a category of models} $\modelcat{\VCcalcolo}$ such that $\VCcalcolo$ \bemph{provides an internal language for these models}.  This is the more interesting and challenging part of the proof and concerns the definition of the notion of morphism of models.  The objects of $\modelcat{\VCcalcolo}$ are models of $\VCcalcolo$ (see Definition~\ref{def:models}), and the suitable definition of morphisms of $\modelcat{\VCcalcolo}$ requires introducing the formal notion of \emph{change of base} (see Defintion~\ref{def:changebase}). 
\commento{Oversimplifying, \red{it} explains how changing the host language one can induce a change of the embedded language.} 

Once $\modelcat{\VCcalcolo}$ has been defined, we known from~\cite{RCSILL} that the typed calculus $\VCcalcolo$ provides an internal language of $\modelcat{\VCcalcolo}$ if we can establish an equivalence $\catvcalcolo\equiv \modelcat{\VCcalcolo}$ between the category of models and the category of theories for $\VCcalcolo$.

To show this equivalence, we need to define a pair of functors: the first functor 
\[\freccia{\catvcalcolo}{S}{\modelcat{\VCcalcolo}}\]
assigns to a theory its syntactic category, and the second functor \[\freccia{\modelcat{\VCcalcolo}}{L}{\catvcalcolo}\] assigns to a pair $(\mV,\mC)$ (a model of $\VCcalcolo$, see Definition~\ref{def:models}) its internal language, as a $\VCcalcolo$-theory. 
Following the schema sketched above, we can proceed with technical results.

Since we will work in the context of enriched categories, we first recall some basic notions about them. We refer to~\cite{BCECT,A2CC} for a detailed introduction to the theory of enriched categories, and we refer to~\cite{CCQCLLECT,En-lambda} for recent applications and connections with theoretical computer science.
%For the rest of this section, let $\mV$ be a fixed monoidal category $\mV=(\mV_0,\otimes,I,a,l,r)$, where $\mV_0$ is a category, $\otimes$ is the tensor product, $I$ is the \emph{unit} object of $\mV_0$, $a$ defines the associativity isomorphism and $l$ and $r$ define the left and right unit isomorphism respectively.  

% where $\mV$ is a cartesian closed category, and $\mC$ is a $\mV$-symmetric monoidal category. 

Let $\mV$ be a fixed monoidal category $\mV=(\mV_0,\otimes,I,a,l,r)$, where $\mV_0$ is a (locally small) category, $\otimes$ is the tensor product, $I$ is the \emph{unit} object of $\mV_0$, $a$ defines the associativity isomorphism and $l$ and $r$ define the left and right unit isomorphism, respectively.

%\begin{def}[Enriched Category]
An \emph{enriched category} $\mV$-\bemph{category}  $\mC$ consists of a class $\ob{\mC}$ of \bemph{objects}, a \bemph{hom-object} $\mC(A,B)$ of $\mV_0$ for each pair of objects of $\mC$, a \bemph{composition law}
\[\freccia{\mC(B,C)\otimes \mC(A,B)}{c_{ABC}}{\mC(A,C)}\]
for each triple of objects, and an \bemph{identity element} 
\[\freccia{I}{j_A}{\mC(A,A)}\]
for each object, subject to associativity and identity laws, see \cite{BCECT,A2CC}.  
%\end{def}

Much of category theory can be reformulated in the enriched setting, such as the notion of monoidal category we will consider here.
Taking for example $\mV=\Set, \Cat, \mathbf{2}, \Ab$ one can re-find the classical notions of (locally small) ordinary category, 2-category, pre-ordered set, and additive category. 
The class of $\mV$-categories forms a category denoted by $\vcat{\mV}$, whose morphisms are $\mV$-functors.

Recall that given $\mA$ and $\mB$ two $\mV$-categories, a $\mV$-\bemph{functor} $\freccia{\mA}{F}{\mB}$ consists of a function
\[\freccia{\ob{\mA}}{F}{\ob{\mB}}\]
together with, for every pair $A,B\in \ob{\mA}$, a morphism of $\mV$ 
\[\freccia{\mA(A,B)}{F_{AB}}{\mB(FA,FB)}\]
subject to the compatibility with the composition and with the identities. Again we refer to \cite{BCECT} and Appendix~\ref{sec:appedix} for the details.

%The first We provide now the definition of the category of models $\modelcat{\VCcalcolo}$, whose objects will be 
 %be pairs $(\mV,\mC)$.
As a first step, we can provide a model for $\VCcalcolo$ in terms of an instance of an \emph{enriched category}, by interpreting the host part into a category $\mV$ and the core part into a suitable category $\mC$ enriched in $\mV$.

%\red{Toglierei questa frase: The models of $\VCcalcolo$ are easily defined as pairs $(\mV,\mC)$ and represent the objects of the category $\modelcat{\VCcalcolo}$ we define below.}

\begin{defi}[Models of $\VCcalcolo$]\label{def:models}
A \bemph{model} of $\VCcalcolo$ is a pair $(\mV,\mC)$ where $\mV$ is a cartesian closed category and $\mC$ is a $\mV$-symmetric monoidal category. 
\end{defi}

\textbf{Notation:} For the rest of this section let $\mV$ be a cartesian closed category and let $\mC$ be a $\mV$-symmetric monoidal category.
%Intuitively, we are interpreting the host language in a cartesian closed category $\mV$ and the core language in a category $\mC$ that is enriched in $\mV$. 

%we do not have any adjunction or comonad between the categories that provide denotation to the host and the core components but, as we will see in Section 3.1, just an instance of an enriched category.

\subsection{Structures and Interpretation}\label{sec:struct}

%\blue{Riformulazione per Ref 1: A \bemph{structure} for the language $\VCcalcolo$ in $(\mV,\mC$) is defined as follows. We assume being given objects $\interp{X}\in \ob{\mV}$ to interpret basis $\host$-types and objects $\interp{A}\in \ob{\mC}$ to interpret basis $\core$-types. Other types are recursively defined in Figure \ref{fig:typesem}.}
A \bemph{structure} for the language $\VCcalcolo$ in $(\mV,\mC$) is specified by giving an object $\interp{X}\in \ob{\mV}$ for every base $\host$-type $X$, and an object $\interp{A}\in\ob{\mC}$ for every base $\core$-type $A$. Given these assignments, the other types are interpreted recursively as
presented in Figure \ref{fig:typesem}:
\begin{figure}[H]
\centering
%\small
\renewcommand\arraystretch{2}
\begin{tabular}{ccc}
$\interp{1}=1$&$\interp{X \times Y}=\interp{X}\times \interp{Y}$&$\interp{X\rightarrow Y}=\interp{X}\rightarrow \interp{Y}$\\
$\interp{\Ctype{\ptype}{A}{B}}=\Ctype{\mC}{\interp{A}}{\interp{B}}$&$\interp{A\otimes B}= \interp{A}\otimes \interp{B}$&$\interp{I}=I$
\end{tabular}
\caption{Type interpretation}\label{fig:typesem}
\end{figure}
\textbf{Notation:} given a $\host$-context $\Gamma=[x_1:X,\dots,x_n:X_n]$, we define $\interp{\Gamma}:=\interp{X_1}\times \cdots \times \interp{X_n}$. Similarly given a $\core$-context $\Omega=[a_1:A_1,\dots,a_m:A_m]$ we define $\interp{\Omega}:=\interp{A_1}\otimes \cdots \otimes \interp{A_m}$.

The interpretation of an $\host$-terms in context $\termconstructor{\Gamma}{\host}{s}{X}$
is given by a morphism of $\mV$
\[\freccia{\interp{\Gamma}}{\interp{\termconstructor{\Gamma}{\host}{s}{X}}}{\interp{X}}\]
defined by induction on the structure of terms as usual. This means, for example, that we define
\[\interp{\termconstructor{x_1:X_1,\dots,x_n:X_n}{\host}{x_i}{X_i}}=\pr_i\]
where $\freccia{\interp{X_1}\times \cdots \times \interp{X_n}}{\pr_i}{\interp{X_i}}$ is the i-\textit{th} projection.

The interpretation of a $\core$-term in context
$\terconstrmix{\Gamma}{\Omega}{\core}{f}{A}$
is given by a morphism 
\[\freccia{\interp{\Gamma}}{\interp{\terconstrmix{\Gamma}{\Omega}{\core}{f}{A}}}{\Ctype{\mC}{\interp{\Omega}}{\interp{A}}}\]
defined again by induction on structure of terms.

Then the variable interpretations for $\core$-judgements are given by
\[\interp{\terconstrmix{\Gamma}{a:A}{\core}{a}{A}}:=
\xymatrix{
\interp{\Gamma} \ar[r]^{!} & 1 \ar[r]^-{j_{\interp{A}}} & \mC(\interp{A},\interp{A}).
}\]
Moreover the terms constructed using promotion and dereliction are interpreted as follows: 
%\red{nuove interpretazioni}
\[\interp{\termconstructor{\Gamma}{\host}{\promote{a_1,\dots a_n.f}}{\Ctype{\ptype}{\otimes_{\Omega}}{A}}}:=\interp{\terconstrmix{\Gamma}{\Omega}{\core}{f}{A}}\]

\[\interp{\terconstrmix{\Gamma}{\Omega}{\core}{\derelict{t,f}}{B}}=\xymatrix@-1pc{\interp{\Gamma}\ar[rr]^-{\angbr{\interp{t}}{\interp{f}}}& &\mC(\interp{A},\interp{B})\times\mC(\interp{\Omega},\interp{A})\ar[rr]^-{c_{\interp{\Omega}\interp{A}\interp{B}}} && \mC(\interp{\Omega},\interp{B}).}\]

\begin{defi}
A structure on $(\mV,\mC)$ \bemph{satisfies} a $\host$-equation in context $\termconstructor{\Gamma}{\host}{s=t}{X}$ if 
$\interp{\termconstructor{\Gamma}{\host}{s}{X}}=\interp{\termconstructor{\Gamma}{\host}{t}{X}}$.
Similarly, a structure satisfies a $\core$-equation in context $\terconstrmix{\Gamma}{\Omega}{\core}{g=f}{A}$ if 
$\interp{\terconstrmix{\Gamma}{\Omega}{\core}{g}{A}}=\interp{\terconstrmix{\Gamma}{\Omega}{\core}{f}{A}}$.
\end{defi}
\begin{rem}
\label{remark dualities of prom and der are satisfied}
As a direct consequence of the previous definitions, we have that 
\[\interp{\terconstrmix{\Gamma}{a:A}{\core}{\derelict{t,a}}{B}}=\interp{\termconstructor{\Gamma}{\host}{t}{\Ctype{\ptype}{A}{B}}}\]
because the second component of the arrow  

\[\interp{\terconstrmix{\Gamma}{a:A}{\core}{a}{A}}:=
\xymatrix{
\interp{\Gamma} \ar[r]^{!} & 1 \ar[r]^-{j_{\interp{A}}} & \mC(\interp{A},\interp{A}).
}\]
acts as the identity with respect to the enriched composition.
Therefore, we can conclude that 
\[\interp{\termconstructor{\Gamma}{\host}{\promote{a.\derelict{t,a}}}{\Ctype{\ptype}{A}{B}}}=\interp{\termconstructor{\Gamma}{\host}{t}{\Ctype{\ptype}{A}{B}}}.\]
Similarly, one can check that the other duality equation of promotion and dereliction is satisfied by our enriched models.
\end{rem}

\begin{rem}
Consider two $\host$-judgements $\termconstructor{\Gamma}{\host}{s}{\Ctype{\ptype}{A}{B}}$ and  $\termconstructor{\Gamma}{\host}{t}{\Ctype{\ptype}{B}{C}}$. Then the interpretation of the composition defined in Remark \ref{remark formal composition} 
\[\termconstructor{\Gamma}{\host}{\enrichedcomp{a.s}{t}}{\Ctype{\ptype}{A}{C}} \]
is given by the arrow of $\mV$ obtained by the following composition

\[\interp{\termconstructor{\Gamma}{\host}{\enrichedcomp{a.s}{t}}{\Ctype{\ptype}{A}{C}} }=\xymatrix@-1pc{\interp{\Gamma}\ar[rr]^-{\angbr{\interp{t}}{\interp{s}}}& &\mC(\interp{B},\interp{C})\times\mC(\interp{A},\interp{B})\ar[rr]^-{c_{\interp{A}\interp{B}\interp{C}}} && \mC(\interp{A},\interp{C}).}\]
Notice that in the interpretation the dependency on the  term $\terconstrmix{\Gamma}{a:A}{\core}{a}{A}$ disappears because, by definition, the second component of the arrow  
\[\interp{\terconstrmix{\Gamma}{a:A}{\core}{a}{A}}:=
\xymatrix{
\interp{\Gamma} \ar[r]^{!} & 1 \ar[r]^-{j_{\interp{A}}} & \mC(\interp{A},\interp{A}).
}\]
acts as the identity with respect to the enriched composition.
\end{rem}

\begin{exa}[Breaking the linearity of the core, Part II]

If we want to extend $\core$ to a non-linear core, as described in Example~\ref{ex:linpart1},  we have to force  the tensor product of the enriched category $\mC$ to be a $\mV$-cartesian product. Notice that the new product is a particular case of the $\mV$-tensor product.

\end{exa}

\subsection{Relating Syntax and Semantics (I): Soundness and Completeness}\label{sec:soundcompl}
We prove the usual relationship between syntax and semantics of $\VCcalcolo$ by means of soundness and completeness theorems.

\begin{thm}[Soundness and Completeness]\label{theorem soundness and completeness}
Enriched symmetric, monoidal categories on a cartesian closed category are sound and complete with respect to the  calculus $\VCcalcolo$.
\end{thm}
\begin{proof}
The soundness for the \emph{pure} $\core$-judgements and \emph{pure} $\host$-judgements is standard, because $\mV$ is cartesian closed, and $\mC$ is symmetric monoidal with hom-sets $\mC_0(A,B)=\mV(I,\mC(A,B))$. What remains to be proven is that the equations in context satisfied by the structure on $(\mV,\mC)$ are closed under the mixed rules, but this follows directly by Remark~\ref{remark dualities of prom and der are satisfied} and by the coherence of the $\mV$-functors $-\otimes -$ and $I\otimes -$. It is a direct generalization of the non-enriched case.

To get completeness, starting from a $\VCcalcolo$-theory $\theory$, we build a cartesian closed category $\mV_{\theory}$ whose objects are $\host$-types and whose morphisms are $\termconstructor{y:Y}{\host}{s}{X}$, both modulo the corresponding equalities. The $\mV_{\theory}$- category $\mC_{\theory}$ is given by the $\core$-objects and the enrichment is given by setting $\mC_{\theory}(A,B):=\Ctype{\ptype}{A}{B}$ and by the mixed judgements together with the corresponding equalities.  
For example the composition morphism
\[\freccia{\mC_{\theory}(B,C)\times \mC_{\theory}(A,B)}{\compfun_{ABC}}{\mC_{\theory}(A,C)}\]
is given by 

{
\begin{prooftree}
%\AxiomC{$\termconstructor{y:\mC(B,C),y: \mC(A,B)}{\mV}{x}{\mC(B,C)\times \mC(A,B)}$}
\AxiomC{$\termconstructor{\Gamma}{\host}{x}{\Ctype{\ptype}{A}{B}}$}
\AxiomC{$\terconstrmix{\Gamma}{a:A}{\core}{a}{A}$}
\BinaryInfC{$\terconstrmix{\Gamma}{a:A}{\core}{\derelict{x,a}}{B}$}
%\AxiomC{$\termconstructor{x:\mC(B,C)\times \mC(A,B)}{\mV}{x}{\mC(B,C)\times \mC(A,B)}$}
\AxiomC{$\termconstructor{\Gamma}{\host}{y}{\Ctype{\ptype}{B}{C}}$}
\BinaryInfC{$\terconstrmix{\Gamma}{a:A}{\core}{\derelict{y,\derelict{x,a}}}{C}$}
\UnaryInfC{$\termconstructor{\Gamma}{\host}{\promote{a.\derelict{y,\derelict{x,a}}}}{\Ctype{\ptype}{A}{C}}$}
\end{prooftree}}
\noindent
(where $\Gamma:= [y:\Ctype{\ptype}{B}{C},x: \Ctype{\ptype}{A}{B}]$) which is exactly what we denote by $\enrichedcomp{a.x}{y}$ in Remark  \ref{remark formal composition}, and it is associative by Remark \ref{remark associativity of composition}. Similarly the identity element
\[\freccia{1}{j_A}{\mC_{\theory}(A,A)}\]
is given by
\begin{prooftree}
\AxiomC{$\terconstrmix{-}{a:A}{\core}{a}{A}$}
\UnaryInfC{$\termconstructor{-}{\host}{\promote{a.a}}{\Ctype{\ptype}{A}{A}}$}
\end{prooftree}
\end{proof}

\subsection{Relating Syntax and Semantics (II): $\VCcalcolo$ is an Internal Language}\label{sec:soundcompl} \mbox{}

In order to prove the stronger result about the correspondence between the category of theories and the category of models, we need to define a suitable notion of morphisms for the latter.

Recall that a \bemph{cartesian closed functor} is a functor preserving finite products and exponents up to isomorphisms, and it is said \bemph{strict} cartesian closed functor if it preserves such structures on the nose.

A key notion which will play a central role in the definition of the category of models is the notion of \emph{change of base}.
This notion explains formally how by changing the host language one can induce a change of the embedded language.
All the definitions and remarks we are going to state work  in a more general context as well, see \cite{BCECT}, but we state them just for our case of interest.

%\begin{defi}[Change of base]\label{def:changebase}
Given two cartesian closed categories $\mV_1$ and $\mV_2$, a cartesian closed functor between them
\[\freccia{\mV_1}{F}{\mV_2}\]
induces a functor 
\[\freccia{\vcat{\mV_1}}{F_{\ast}}{\vcat{\mV_2}}\]
sending a $\mV_1$-category $\mC$ to the $\mV_2$-category $F_{\ast}(\mC)$ such that
\begin{itemize}
\item $F_{\ast}(\mC)$ has the same objects as $\mC$;
\item for every $A$ and $B$ of $\mC$, we define
\[F_{\ast}(\mC)(A,B)=F(\mC(A,B))\]
\item the composition, unit, associator and unitor morphisms in $F_{\ast}(\mC)$ are the images of those of $\mC$ composed with the structure morphisms of the cartesian closed structure on $F$.
\end{itemize}
See \cite[Lemma 3.4.3]{CHT} for all the details.
\begin{rem}
Notice that if $\freccia{\mV_1}{F}{\mV_2}$ and $\freccia{\mV_2}{G}{\mV_3}$ are cartesian closed functors, then the functor $\freccia{\vcat{\mV_1}}{(G\circ F)_{\ast}}{\vcat{\mV_3}}$  coincides with the functor $G_{\ast}\circ F_{\ast}$. In fact, it is direct to check that for every $\mV_1$-category $\mC$, the categories $(G\circ F)_{\ast}(\mC)$ and $G_{\ast}(F_{\ast}(\mC))$ have the same objects and, for every $A$ and $B$ of $\mC$, we have that
\[(G\circ F)_{\ast}(\mC)(A,B)=(G\circ F)( \mC (A,B))= G(F(\mC(A,B))=(G_{\ast}\circ F_{\ast}) (\mC)(A,B). \]
Similarly, it is direct to check that the functor $\freccia{\vcat{\mV}}{(\mathit{Id}_{\mV})_{\ast}}{\vcat{\mV}}$ induced by the identity functor on $\mV$ coincides with the identity functor on $\vcat{\mV}$.
\end{rem}
\begin{defi}[Change of base]\label{def:changebase}
Given a cartesian closed functor $\freccia{\mV_1}{F}{\mV_2}$, the induced functor \freccia{\vcat{\mV_1}}{F_{\ast}}{\vcat{\mV_2}} is called \bemph{change of base}.
\end{defi}
A standard example is the \emph{underlying set functor} $\freccia{\mV}{\mV(I,-)}{\Set}$, sending an object $A$ of $\mV$ to the underlying set of elements $\mV(I,A)$. This lax-monoidal functor induces a functor $\freccia{\vcat{\mV}}{\mV(I,-)_{\ast}}{\Cat}$, see \cite[Proposition 3.5.10]{CHT} or \cite{BCECT}.

\begin{rem}
Observe that if $\freccia{\mV_1}{F}{\mV_2}$ is a strict cartesian closed functor and $\mC$ is a $\mV_1$-symmetric monoidal category, then $F_{\ast}(\mC)$ is a $\mV_2$-symmetric monoidal category. 
\end{rem}

We exploit the equivalence between the category of theories and the category of models to show, from a type theoretical point of view, that the change of base has the following meaning: if one has a translation between two host languages $\host_1$ and $\host_2$ and $\core$ is an embedded language in $\host_1$, then one can use this translation of host languages to \emph{change the host language} for $\core$, obtaining an embedding in $\host_2$.

Now we can define the category of models $\modelcat{\VCcalcolo}$.
\begin{defi}[Category of $\VCcalcolo$ models]
The objects of $\modelcat{\VCcalcolo}$ are models in the sense of Definition \ref{def:models}, and an arrow in $\modelcat{\VCcalcolo}$ between two models $(\mV_1,\mC_1)$ and $(\mV_2,\mC_2)$ is given by a pair $(F,f)$ where $\freccia{\mV_1}{F}{\mV_2}$ is a strict cartesian closed functor, and $\freccia{F_{\ast}(\mC_1)}{f}{\mC_2}$ is a $\mV_2$-functor, preserving strictly all the structures. Given two arrows $\freccia{(\mV_1,\mC_1)}{(F,f)}{(\mV_2,\mC_2)}$ and $\freccia{(\mV_2,\mC_2)}{(G,g)}{(\mV_3,\mC_3)}$ the composition of these is given by the pair $(G,g)(F,F):=(GF,\ovln{gf})$ where
\[ \ovln{gf}:=g G_{\ast}(f).\]
\end{defi}
Observe that the requirement that all the functors and enriched functors must be strict reflects the definition of translation of theories.

We can conclude with the theorem showing the equivalence between the categories of models and that of theories for $\VCcalcolo$.

\begin{thm}\label{theorem internal laguage}
The typed calculus $\VCcalcolo$ provides an internal language for $\modelcat{\VCcalcolo}$, i.e. $\catvcalcolo$ is equivalent to $\modelcat{\VCcalcolo}$.
\end{thm}
\begin{proof}

We define a functor $\freccia{\catvcalcolo}{S}{\modelcat{\VCcalcolo}}$ by mapping a $\VCcalcolo$-theory to its syntactic category defined as in Theorem \ref{theorem soundness and completeness}. On morphisms, the functor $S$ takes a translation between theories to the pair of functors induced on the syntactic categories by mapping terms and types to their respective translations.

Conversely, the functor $\freccia{\modelcat{\VCcalcolo}}{L}{\catvcalcolo}$ is defined by mapping an object $(\mV,\mC)$ of $\modelcat{\VCcalcolo}$ to the $\VCcalcolo$-theory obtained by extending $\VCcalcolo$ with:

\begin{itemize}
\item new $\core$-types $A$ with the corresponding axiom $\typeconstructor{\core}{A}$ for each element of $\ob{\mC}$, which means that we are naming the objects of $\mC$ in the calculus and we extend the interpretation of the $\core$-types of $\VCcalcolo$ by interpreting the new names with the corresponding objects;
\item new $\host$-types $X$ with the corresponding axiom $\typeconstructor{\host}{X}$ for each element of $\ob{\mV}$ by naming the objects of $\mV$ as we did for $\mC$, and by renaming the types and the axioms induced by objects of the form $\mC(A,B)$ as $\typeconstructor{\host}{\Ctype{\ptype}{A}{B}}$;
\item new $\core$-terms $f$ and $\terconstrmix{\Gamma}{\Omega}{\core}{f}{A}$ for each morphism $\freccia{\Gamma}{f}{\mC(\Omega,A)}$ of $\mV$ having the interpretation of $\Gamma$ as domain and of $\mC(\Omega,A)$ as codomain;
\item new $\host$-terms $s$ and $\termconstructor{\Gamma}{\host}{s}{X}$ for each morphism $\freccia{\Gamma}{s}{X}$ of $\mV$ having the interpretation of $\Gamma$ as domain and of $X$ has codomain;
\item new equality axioms $\typeconstructor{\core}{A=B}$ if the interpretation of $A$ is equal to that of $B$ in $\mC$;
\item  new equality axioms $\typeconstructor{\host}{X=Y}$ if the interpretation of $X$ is equal to that of $Y$ in $\mV$;
\item  new equality axioms between $\core$-terms $\terconstrmix{\Gamma}{\Omega}{\core}{f=g}{A}$ if the interpretation of $\terconstrmix{\Gamma}{\Omega}{\core}{f}{A}$ is equal to that of $\terconstrmix{\Gamma}{\Omega}{\core}{g}{A}$ by interpreting the new term symbols in the morphisms they name;
\item  new equality axioms between $\host$-terms $\termconstructor{\Gamma}{\host}{s=t}{X}$ if the interpretation of $\termconstructor{\Gamma}{\host}{s}{X}$ is equal to that of $\termconstructor{\Gamma}{\host}{t}{X}$ by interpreting the new term symbols as the morphisms they name. 
\end{itemize}
The morphisms of $\modelcat{\VCcalcolo}$ give rise to translations, because both the components of morphisms are strict. It is direct to check that the previous functors define an equivalence of categories.
\end{proof}

\begin{rem}
Observe that functors defined in Theorem \ref{theorem internal laguage} define the isomorphisms $(\mV,\mC) \cong S(L(\mV,\mC))$ and $\theory\cong L(S(\theory))$. 

Some other authors consider only the equivalence $(\mV,\mC) \equiv S(L(\mV,\mC))$, see for example \cite{CTCS}, as the characterizing property of the internal language, however in this case one does not obtain the equivalence between the categories of models and that of theories, but just a sort of bi-equivalence. \commento{questa nota sulla bi-equiv si pu\'o dire un po' meglio? a sort of \`e un po' informale}
\end{rem}
We conclude this section with some concrete examples of $\VCcalcolo$-models.
\begin{exa}\label{ex:concreteex}\mbox{}
\begin{enumerate}
\item Every locally small symmetric monoidal category $\mM$ is by definition $\Set$-enriched, and therefore gives rise to a $\VCcalcolo$-model given by the pair $(\Set,\mM)$.

\item The category of $\cstar$-algebras provides a concrete example of model in the form $(\Set, \mM)$. Recall  that a (unital) $\cstar$-algebra is a vector space over the field of
complex numbers with multiplication, unit and involution. Moreover, it satisfies
associativity and unit laws for multiplication, involution laws and  the spectral radius provides a norm making it
a Banach space~\cite{CAWA,RENNELA2018257,CCQCLLECT}. Finite dimensional $\cstar$-algebras and
completely positive unital linear maps form a symmetric monoidal category denoted by $\fdcstar$. See~\cite{RENNELA2018257,CCQCLLECT}.
Therefore the pair $(\Set,\fdcstar)$ is a model of $\VCcalcolo$.

\item
The pair $(\DCPO_{\bot},\REL)$ provides a model for $\VCcalcolo$, where $\REL$ is the category of sets and relations. Recall that $\DCPO_{\bot}$ is cartesian closed and $\REL$ is enriched in $\DCPO_{\bot}$ by the usual subsets ordering, and the least element of each hom-set is given by the empty relation. This pair is particularly interesting when one plans to move toward quantum languages (with classical control~\cite{Zorzi16}). In Section \ref{sec:travail} we develop this example to show some possible applications of Theorem~\ref{theorem internal laguage}.
\end{enumerate}
\end{exa}
\iffalse
In categorical logic tradition the internal language theorem represents ``a step beyond soundness and completeness'' to relate syntax and semantics in a stricter fashion, \cite{RCSILL}.  
 In the Section~\ref{sec:travail} we show that the notion of internal language  is also central and provides a formal characterization of the host-core programming style.% since it allow us to reason equivalently both on the syntax and on the semantics. 
\fi
\section{The host-core ``internal language'' {in action}} \label{sec:travail}
% and allows us to provide a formal definition of host-core language built out from two arbitrary standalone languages $\lang{1}$ and $\lang{2}$.

%In this section we show that $\VCcalcolo$ provides an internal language for the category of its models $\modelcat{\VCcalcolo}$. To prove this we have to define the category of its theories  $\catvcalcolo$, the category $\modelcat{\VCcalcolo}$ and then to prove the equivalence $\catvcalcolo\equiv\modelcat{\VCcalcolo}$.% of the theories of the language~\cite{DMarxiv}.

%This result can be intuitively read as follows: syntactically, $\VCcalcolo$ is the most expressive language having the defined semantics as a model and, semantically, the model is the smallest one providing a denotation for the calculus.

{The equivalence $\catvcalcolo\equiv\modelcat{\VCcalcolo}$ provided in Theorem \ref{theorem internal laguage} can be read as follow: we can reason equivalently either starting from the syntax or from the semantics. 
Thanks to this strong correspondence, we are able to answer the main questions raised in the introduction.}
\subsection{Answering Q1: A formal definition of Host-Core Languages}\label{sec:hcdefinition}

\begin{defi}\label{def:hostcore}
%Let $\lang{1}$ and $\lang{2}$ be two programming languages and $\mathcal{D}_1$ and $\mathcal{D}_2$ their denotations.
%If it is possible to define a model $\mathcal{D}_3$ as the enrichment of $\mathcal{D}_2$ in $\mathcal{D}_1$, the host-core language $\lang{1}-\lang{2}$ is the internal language of the model $\mathcal{D}_3$.

Let $\lang{1}$ and $\lang{2}$ be two standalone languages with their syntactic denotations $D_1$ and $D_2$. We say that $\lang{2}$ is \emph{embeddable} in $\lang{1}$ if 
the pair $(D_1,D_2)$ belongs to $\modelcat{\VCcalcolo}$. 
\end{defi}
Employing the equivalence $\catvcalcolo\equiv\modelcat{\VCcalcolo}$  of Theorem \ref{theorem internal laguage}, and the functor $\freccia{\modelcat{\VCcalcolo}}{L}{\catvcalcolo}$ defined in the proof of such a theorem, we can conclude that if $\lang{2}$ is embeddable in $\lang{1}$ in the sense of the previous definition, then
%, by the previously cited equivalence $\catvcalcolo\equiv\modelcat{\VCcalcolo}$ 
 there exists a  host-core language $\lang{3}:=L(D_1,D_2)$ such that $\lang{1}$ hosts $\lang{2}$. The requirement that the denotation of $\lang{2}$ is enriched in the denotation of $\lang{1}$ from a syntactic point of view means that there exists a mixed language whose ``host'' part is $\lang{1}$ and whose ``pure'' core part is $\lang{2}$. This coincides with the intuitive notion one could have of embedded style languages.

Notice that this semantic definition of the host-core languages confirms what we explain in the introduction. We are not defining a mere extension of $\lang{1}$, but we want to describe the fact that $\lang{1}$ is able to host $\lang{2}$ and provide the interface between the two languages. This is syntactically characterized through the mixed syntax, and semantically captured by the enrichment.  
\subsection{Answering Q2: Relating the standalone and the host-core style, a first step}\label{sec:equivsahc}
The equivalence $\catvcalcolo\equiv\modelcat{\VCcalcolo}$ provides a useful tool for comparing languages presented with different features and styles, as standalone and host-core languages. Indeed, without these categorical results, it would be hard to find the right criteria to compare such  different languages, while lifting this problem in the categorical settings and using the correspondence between theories and models, a comparison can be easily done just by studying their categorical models.

Consider Benton's presentation \LNL\ of Intuitionistic Linear Logic (\texttt{ILL}). In~\cite{RCSILL} the authors proved that $\modelcat{\LNL}$ is a full subcategory of 
$\modelcat{\texttt{ILL}}$.
This means that a ``mixed grammar'' in the style of \LNL\ is equivalent to another one written in \texttt{ILL} in a ``standalone'' style, in the sense that they have the same models.
%QUESTA LA TOGLIEREI Mapping this idea to programming languages, we are able to express and compare all programs written in  $\lang{2}$ in the syntax of $\lang{1}$.
{Similarly, we are able to compare }the pure linear ({exponential-free}) fragment of $\texttt{ILL}$, called here $\texttt{RLL}$ as in~\cite{RCSILL}, with $\VCcalcolo$. In fact, we can show that $\modelcat{\texttt{RLL}}$ is a subcategory of $\modelcat{\VCcalcolo}$. The objects of $ \modelcat{\texttt{RLL}}$ are small symmetric monoidal categories and then, by definition, they are $Set$-enriched. Each morphism of $\modelcat{\texttt{RLL}}$, which by definition preserves all appropriate structure, can be extended to a morphism of $\modelcat{\VCcalcolo}$. This morphism is 
of the form $({Id},m)$.

\subsection{Answering Q3: From semantics back to syntax, an empirical example}\label{sec:extractionalgo}

The functors $\freccia{\modelcat{\VCcalcolo}}{L}{\catvcalcolo}$ and $\freccia{\catvcalcolo}{S}{\modelcat{\VCcalcolo}}$ defined in the proof of Theorem~\ref{theorem internal laguage} provide a useful tool to define a $\VCcalcolo$-theory $L(\mV,\mC)$ starting from a pair $(\mV,\mC)$, where $\mV$ is a cartesian closed category and $\mC$ is a $\mV$-symmetric monoidal category, and vice versa. Here we focus on the former direction.
In particular, the explicit way in which the $\VCcalcolo$-theory $L(\mV,\mC)$ is defined from a given model $(\mV,\mC)$ suggests that the functor $\freccia{\modelcat{\VCcalcolo}}{L}{\catvcalcolo}$ can be used to reason backwards, from semantics to syntax.
%TOLTO QUESTE 4 RIGHE
%starting from categories of interest, \blue{through logical specification of the proof of the Internal language theorem}, we  obtain a \emph{candidate}, mixed, minimal $\VCcalcolo$-theory and we can use it as a basis for programming language design. 
%Notice that, as a byproduct, we can also get the separate syntaxes of the two communicating languages (denoted by the two input categories).
{ In Example \ref{ex:dcpobrel} we sketch how to extract a basis kernel type theory for a quantum language with classical control starting from a model (two categories of interest). What we obtain can be used as a starting point to fully define the syntax and the well-typing rules. Starting from the semantics one gets a kernel set of syntactical constructors mirroring the denotation.
This seems to be useful in non-classical computations, and in particular in quantum languages, where the definition of suitable semantics is not always trivial at all  (see e.g. ~\cite{CVW19}). }

\begin{exa}\label{ex:dcpobrel}
We apply the functor $\freccia{\modelcat{\VCcalcolo}}{L}{\catvcalcolo}$ to the pair of models $(\DCPO_{\bot}, \Rel)$, where $\Rel$ is enriched in $\DCPO_{\bot}$. The pair $(\DCPO_{\bot}, \Rel)$ is of interest since it captures a communication with (the denotation of) a ``classical'' computational system represented by the $\DCPO_{\bot}$ with (the denotation of)  another system having a ``quantum flavour'', since $\Rel$ is a dagger category in the sense of Selinger (see e.g.,~\cite{Selinger07a}).
%We choose this example since $Rel$ is a well known in questo caso possiamo lavorare sint e sem in modo equivalente

The $\VCcalcolo$-theory $L(\DCPO_{\bot}, \Rel)$ is obtained by extending $\VCcalcolo$ with:
\begin{itemize}
\item  a new $\core$-type $A$ for every object $A$ of $\Rel$;

\item a new $\host$-type $X$ for every object $X$ of $\DCPO_{\bot}$; moreover, for each $A,B\in\Rel$, we add a new $\host$-type ${\Ctype{\ptype}{A}{B}}$ which is the set of binary relations between  $A$ and $B$;

%\item \textbf{Proof Types:} Given two sets $A$, $B$ the set of ${\Ctype{\ptype}{A}{B}}$ is the set of \emph{binary relations} between $A$ and $B$.

\item a new $\core$-term $f$ for each morphism $f:\Gamma\rightarrow \Rel(\Omega,A)$ of $\DCPO_{\bot}$;% we introduce the $\core$-term $f$ 
%and its judgement $\terconstrmix{\Gamma}{{\Omega}}{\core}{f}{A}$ 

\item a new $\host$-term $s$ for each morphism $s:\Gamma\rightarrow X$ of $\DCPO_{\bot}$;
%We introduce the $\host$-term $s$.% and its judgement $\termconstructor{\Gamma}{\host}{s}{X}$.
\item new equality axioms $\typeconstructor{\core}{A=B}$ if the sets $A$ and $B$ are equal sets;
\item  new equality axioms $\typeconstructor{\host}{X=Y}$ if $X$ and $Y$ are equal objects of  $\DCPO_{\bot}$;
\item  new equality axioms between $\core$-terms $\terconstrmix{\Gamma}{\Omega}{\core}{f=g}{A}$ if 
$f:\Gamma\rightarrow \Rel(\Omega,A)$ and $g:\Gamma\rightarrow \Rel(\Omega,A)$ are equal morphisms of $\DCPO_{\bot}$;

%the interpretation of $\terconstrmix{\Gamma}{\Omega}{\core}{f}{A}$ is equal to that of $\terconstrmix{\Gamma}{\Omega}{\core}{g}{A}$ by interpreting the new term symbols in the morphisms they name;
\item  new equality axioms between $\host$-terms $\termconstructor{\Gamma}{\host}{s=t}{X}$ if $s:\Gamma\rightarrow X$ is equal to  $t:\Gamma\rightarrow X$ as morphisms of $\DCPO_{\bot}$. 

%\item $\cjudset$ contains  judgements of the form $\terconstrmix{\Gamma}{{\Omega}}{\core}{f}{A}$ for each $f\in\ctermset$ and an equality judgement $\terconstrmix{\Gamma}{\Omega}{\core}{f=g}{A}$ for each pair of core terms $f$ and $g$ having the same denotation.
%
%\item $\hjudset$ contains  judgements of the form $\termconstructor{\Gamma}{\host}{s}{X}$ for each $s\in\htermset$ and an equality judgement $\termconstructor{\Gamma}{\host}{s=t}{X}$ for each pair of host terms $s$ and $t$ having the same denotation.
\end{itemize}
\end{exa}
%Notice that given a core term $f$, the corresponding ``promoted'' host term is defined as $\termconstructor{\Gamma}{\host}{\promote{f}}{\Ctype{\ptype}{\Omega}{A}}$ and 
This example highlights the intuition behind the notion of enrichment. In few words, we can speak about relations between sets in the context of $\DCPO_{\bot}$. This is particularly evident if we consider pure core judgements (i.e., defined in a contex whose host part is empty). 
For each object in $\Rel(A,B)$ we have in $\DCPO_{\bot}$ a function defined from the initial object 1 to $\Rel(A,B)$, i.e. $\termconstructor{}{\host}{\promote{a.f}}{\Ctype{\ptype}{A}{B}}$ 
can be viewed as a constant and we can plainly use constants, that are ``silent programs'' in the host syntax.

Since $\Rel$ is a dagger category,  in this syntax we can explicitly specify the term constructor $\dag$. For example:

\begin{prooftree}
\AxiomC{$\termconstructor{\cdot \;|\; a:A}{\core}{{f}}{B}$}
\UnaryInfC{$\termconstructor{\cdot\;|\; b:B}{\core}{{f}^{\dag}}{A}$}
\end{prooftree}
We can require that the operator satisfies all the expected equations~\cite{Selinger07a} and so model this through syntactical rules.
The real advantage of such a $\VCcalcolo$-theory is that one can equivalently reason about its term judgements both syntactically or semantically, viewing them as a concrete functions. We plan to develop the example above and use the set of types, terms and judgements we extracted as a basis for quantum language design (see our future work, Section~\ref{sec:conclusions}).

%% trovare un linguaggio standalone abbastanza espressivo per ``coprire'' un host-core

\section{Related work, Discussions and Future Work}\label{sec:related}

\subsection{Related Work}

%\commento
{
Benton's   Linear-Non-Linear Logic provides an elegant presentation of Linear Logic~\cite{LLLC},  and in the last years its models have been studied both from a categorical logic and a computer science perspective.

An important contribution in the first direction has been given by Maietti et al. in \cite{RCSILL}, where the authors discuss models and morphisms for Intuitionistic Linear Logic (\texttt{ILL}), Dual Intuitionistic Linear Logic (\texttt{DILL}) and Linear-Non-Linear Logic (\LNL). 
The crucial point is that soundness and completeness theorems are not generally sufficient to identify the most appropriate class of denotational models for a typed calculus, unless the same typed calculus provides an internal language of the category of models we are considering.

In the context of the foundation of programming theory, the language $\ewire$ is studied in~\cite{RENNELA2018257} as the basis of a denotational semantics based on enriched categories. It is built from a simple first-order linear language for circuits embedded in a more powerful host language.
The circuit language is interpreted in a category that is enriched in the category denoting the host part. 
Moreover, some interesting extensions of the host language are proposed. In particular, the authors use the enrichment of the category of W*-algebras in the $\DCPO$-category to accommodate recursive types. This allows them  to model arbitrary types and is directly connected with the possibility of easily encoding \emph{parametric} quantum algorithms.
In~\cite{CCQCLLECT} the authors also show a relation with Benton's Linear-Non-Linear models.

{An inspiring  work both for $\VCcalcolo$ and  $\ewire$ is the \emph{enriched effect calculus} (\texttt{EEC})~\cite{EggerMS14}, whose models are given in terms of enriched categories. The idea behind the semantics in \cite{RENNELA2018257} and in \cite{EggerMS14} is quite similar. Also in \cite{EggerMS14} a deep comparison with \LNL\ models is provided, showing that every \LNL\ model with additives determines an (\texttt{EEC}) model. The authors also prove soundness and completeness of the equational theories, with respect to the interpretation, but they do not consider the notion of internal language.}

Looking for non-classical computation, $\ewire$ has been defined as a generalization of a version of $\qwire$ (\emph{``choir''})~\cite{qwire17a,qwire17b,qwire3}, which is one of the most advanced programming platforms for the encoding and the verification of quantum circuits. 
  
The \emph{circuit language} of $\qwire$ can be treated as  the quantum plugin for the host classical language, currently in the Coq proof assistant~\cite{qwire17b,RobRandThesis}. The  type system is inspired by Benton's (\LNL) Logic and supports both linear and dependent types. % that partitions the exponential data into a purely linear fragment and a purely non-linear fragment connected via a categorical adjunction (notice that this fully reflects also the QRAM architecture). %This choice also makes the ``quantum core'' strongly independent from the host language. 
The circuit language essentially controls the well-formed expressions concerning  \emph{wires}, i.e. circuit inputs/outputs, 
whereas the host language controls higher-order circuit manipulation and evaluation.
% The type system of the circuit language essentially controls the well formation of expressions concerning  \emph{wires}, i.e. circuit's inputs/outputs. 
%whereas 
The host language also controls the boxing mechanism (a circuit can be ``boxed'' and then promoted as a classical resource/code.  %, see Section~\ref{sec:discussions}.
%$\qwire$  also supports features such as  \emph{dynamic} and \emph{static} lifting.
See~\cite{RobRandThesis} for a complete account about the use of advanced operations and techniques designed for $\qwire$. 
Further developments on this line are provided in  \cite{En-lambda}, where the authors introduce the lambda calculus \texttt{ECLNL} for string diagrams, whose primary purpose is to generate complicated diagrams from simpler components. In particular, the language \texttt{ECLNL} adopts the syntax (and operational semantics) of \texttt{Proto-Quipper-M}, while  the categorical model is again given by a \LNL\ model, but endowed with an additional enrichment structure. The  abstract model of \texttt{ECLNL} satisfies the soundness, while  completeness and internal language are not discussed in \cite{En-lambda}. However, notice that \texttt{ECLNL} can be seen as a particular specialization of $\VCcalcolo$, and also its categorical model is a particular case of the model we introduced. This is not surprising, since one of the main motivation for the design of $\VCcalcolo$ claimed at the beginning of our work is that $\VCcalcolo$ has to  embody the principal properties of those languages dealing with host-core situations.

A notion of \emph{multi-language}  related to $\VCcalcolo$  is defined and studied in~\cite{BuroM19,BuroCM20} on the basis of the pioneering investigation~\cite{MatthewsF07}. 
The authors address the problem of providing a formal semantics to the combination of programming languages by  introducing an algebraic framework based on \emph{order-sorted algebras}. The framework provides an abstract syntax (induced by the algebraic structure), that works regardless of the inherent nature of the combined languages. 
While the notion of multi-language they introduced in \cite{BuroM19,BuroCM20} seems  similar to the one we propose in this paper,  motivations, syntactic features and  semantics are deeply different. First, the authors consider two languages combined in a unique standalone language essentially given by the union of the two syntaxes (without a``mixed zone''). Thus,  two languages can be considered at the same level, i.e. there is not a language that has a \emph{privileged position} with respect to the other. This is in contrast with our formulation, which models situations in which the host can delegate some computations to the core. %, which means that the second language is supporting the main one. 
The absence of hierarchy between syntaxes composing the multi-language framework defined in~\cite{BuroM19,BuroCM20}  is reflected also semantically: given two signatures $S_1$ and $S_2$,  the  categorical model of the multi-language obtained from them is given essentially by a pair of algebras $A_1$ and $A_2$, where the first is an algebra over $S_1$, and the second one is an algebra over $S_2$.   No enrichment is required, since, as just said, in the resulting multi-language the set of sorts is obtained by the union of the sets of sorts of $S_1$ and $S_2$. 
Finally, with respect to \cite{BuroM19,BuroCM20} we provide a type theory and a denotational semantics together with the notion of
internal language, which is to the best of our knowledge, new in the context of host-core languages.

%NUOVE CONCLUSIONI
\subsection{Discussions and Future Work}\label{sec:conclusions}

In this paper, we design the host-core calculus $\VCcalcolo$  and provide its denotational semantics in terms of enriched categories.% Staton et Al.'s previous investigations~\cite{}.

%{Our main references come from the categorical logic literature. We also toke inspiration from the state of the art of \emph{embedded} programming (especially from the quantum setting).}

We remark that the point of view we developed in this paper is different  from  the programming language design one,  and  we are aware that  $\VCcalcolo$ ``as is'' can not provide a full theoretical account of host-core programming. %and, in perspective, of multi-language practice. 
 Notwithstanding, we believe our investigation provides the first steps towards a foundation for a particular case of multi-language interaction systems.
We tried to show this throughout the paper, through examples and applications of the notion of internal language.
{We claim that a principled theory, with a pre-existing, extensible library of traditional type-theoretical results, relating the syntax and semantics of the underlying main notions of our host-core system, is the appropriate basis on which to build more complicated, less homogeneous systems, more adapted to the applications at hand.}

A crucial aspect that will be central in any future refinement of $\VCcalcolo$ is the communication between host and core. 
{This fact becomes central when one wants to model a notion of \emph{interoperability}. The issue is out of the scope of our  investigation, {but we quickly comment on the point from a type theory perspective.} }
%\textst{In the current version of $\VCcalcolo$, the communication between  host and core language is limited, but, as displayed in Examples} \textst {in Section, this restricted case allows to make some interesting considerations.}
{The way host $\host$ and core $\core$ communicate with each other is strongly related to the choice one makes to design the host and to the expressive power of $\host$ and $\core$.}

%In the design of $\VCcalcolo$, we establish a hierarchical dependency of the core on the host.
%This is fully reflected by the kind of \emph{communication} we model between $\host$ and $\core$ languages. 
More expressive languages can arise from more complex forms of interaction. At the level of the mixed type theory, a ``realistic'' version of communication able to reflect real language interoperability necessarily requires a notion of \emph{casting} between host and core types.  Some host types $X$ can be  ``isomorphic'' to types of the shape $\Ctype{\ptype}{I}{A}$ of promoted core terms. For example, one could have a boolean type $\boldbool$ and state it is equivalent to a type $\Ctype{\ptype}{I}{\boldbit}$, where $\boldbit$ is the core type of bits.
It is natural to think of adding a constant $\mathbf{cast}:\Ctype{\ptype}{I}{A}\mapsto X$ that allows $\host$ to explicitly read results of the evaluation in $\core$. 
One could consider also a reverse casting implementing a bi-directional notion of interoperability.

\subsubsection{ Future Work}

%%%%%%%%%%%%%%%%%%%
Our investigation leads into (at least) three directions: the improvement of the expressive power of $\VCcalcolo$, a complete operational study of $\VCcalcolo$,  and finally its quantum specialization.

\begin{itemize}
\item We design $\VCcalcolo$ pursuing a notion of compositionality. We consider the system as a kernel calculus both for extensions and specializations.

We aim to use the direct correspondence between syntax and semantics to obtain more refined type theories by adding syntactical rules and (equivalent) denotational properties, without changing the rules of the basic language.
%  $\VCcalcolo$.

Each extension requires the addition of syntactical primitives and, mirroring the syntax, a suitable definition of models $(\mV,\mC)$, where the core category $\mC$ denotes the peculiar features of the paradigm and is enriched in $\mV$. 

\item We plan to  study the operational semantics of $\VCcalcolo$ and its
related safety properties such as Subject Reduction and Progress Theorems, as well as a notion of normal form for the $\VCcalcolo$ computations.
%\textst{In our opinion, this will be the breaking point that will allow us to move our focus towards the definition of a (paradigmatic) host-core functional language.}
The topic is already interesting for the current formulation of the syntax and seems to become challenging if one considers more expressive languages. {We expect that our operational semantics will not be similar to that of $\qwire$ and $\ewire$, since the communication between $\host$ and $\core$ and the evaluation style we designed are different.}
We are also interested in the study of reduction strategies. %\textst{, another task that we are going to address as a medium time goal.}

\item %Among the possible specializations of $\VCcalcolo$, the quantum one, which was our main inspiration in the first phase of this work, is in our opinion the most interesting.
One can build a new quantum  specialization of $\VCcalcolo$. We aim to start from a sufficiently expressive host language $\host$ and a quantum ``tuning'' of a circuit description core in the style of $\core^*$ (Example~\ref{ex:running2}), possibly improving its expressiveness. 
Once we have defined the quantum specialization of $\VCcalcolo$, we plan to compare it with other established and influential languages, such as $\qwire$ and $\ewire$.
%This is a short-medium time goal. 
As a parallel task in this quantum setting, we are interested in further developments of  Example~\ref{ex:dcpobrel}, focusing  on  pairs of enriched categories such as  $(\DCPO_{\bot}, \REL)$. Thanks to the internal language theorem, we hope to reverse the perspective in language design, by extracting new type systems backward from the semantics.

\end{itemize}

%% in general the use of bibtex is encouraged
\bibliographystyle{alpha}
\bibliography{biblio}

\newpage
%%%%%%%%%%%%%%%%%%%%%%%%%%%%%%%%
\appendix
\section{Equational theory of the host languages $\host$}\label{app:evaluation rules}\label{Appendix}
The equations of the host language are the standard one of simply typed lambda calculus: 
\begin{figure}[H]
  \centering
  \renewcommand\arraystretch{3}
  %\small
  \begin{tabular}{c}
  %\AxiomC{$\terconstrmix{\Gamma}{\Omega_1,a:A}{\mC}{f}{B}$}
  %\AxiomC{$\terconstrmix{\Gamma}{\Omega_2}{\mC}{g}{A}$}
  %\BinaryInfC{$\terconstrmix{\Gamma}{\Omega_1,\Omega_2}{\mC}{(\lambda a:A.f)(g)=f[g/a]}{B}$}
  %\DisplayProof \\
  ($\beta$) \AxiomC{$\termconstructor{\Gamma}{\host}{\lambda x:X.s}{X\rightarrow Y}$}
  \AxiomC{$\termconstructor{\Gamma}{\host}{t}{X}$}
  \AxiomC{}
  \TrinaryInfC{$\termconstructor{\Gamma}{\host}{(\lambda x:X.s)t=s[t/x]}{Y}$}
  \DisplayProof \\
  ($\pi_1$)\AxiomC{$\termconstructor{\Gamma}{\host}{\langle s,t\rangle}{X\times Y}$}
  \AxiomC{}
  \BinaryInfC{$\termconstructor{\Gamma}{\host}{{\pi_1\langle s,t\rangle}{}{}=s}{X}$}
  \DisplayProof \\
  ($\pi_2$)\AxiomC{$\termconstructor{\Gamma}{\host}{\langle s,t\rangle}{X\times Y}$}
  \AxiomC{}
  \BinaryInfC{$\termconstructor{\Gamma}{\host}{{\pi_2\langle s,t\rangle}{}{}=t}{Y}$}
  \DisplayProof \\
  ($\eta$)\AxiomC{$\termconstructor{\Gamma}{\host}{\lambda x:X.s}{X\rightarrow Y}$}
  \AxiomC{
  %$\termconstructor{\Gamma}{\host}{x}{X}$
  }
  \AxiomC{}
  \TrinaryInfC{$\termconstructor{\Gamma}{\host}{\lambda x:X.s(x)=s}{Y}$}
  \DisplayProof\\
  ($l.a$)\AxiomC{$\termconstructor{\Gamma}{\host}{s=t}{X}$}
  \AxiomC{$\termconstructor{\Gamma}{\host}{u}{X\rightarrow Y}$}
  \BinaryInfC{$\termconstructor{\Gamma}{\host}{{us=ut}{}{}}{Y}$}
  \DisplayProof \\
  ($r.a$)\AxiomC{$\termconstructor{\Gamma}{\host}{s=t}{X\rightarrow Y}$}
  \AxiomC{$\termconstructor{\Gamma}{\host}{u}{X}$}
  \BinaryInfC{$\termconstructor{\Gamma}{\host}{{su=tu}{}{}}{Y}$}
  \DisplayProof \\
  (in.$\lambda$)\AxiomC{$\termconstructor{\Gamma, x:X}{\host}{s=t}{Y}$}
  \AxiomC{}
  \BinaryInfC{$\termconstructor{\Gamma}{\host}{{\lambda x.s=\lambda x.t}{}{}}{X\rightarrow Y}$}
  \DisplayProof \\
  (refl)\AxiomC{$\termconstructor{\Gamma}{\host}{s}{X}$}
  \AxiomC{}
  \BinaryInfC{$\termconstructor{\Gamma}{\host}{{s=s}{}{}}{X}$}
  \DisplayProof \\
  (sym)\AxiomC{$\termconstructor{\Gamma}{\host}{s=t}{X}$}
  \AxiomC{}
  \BinaryInfC{$\termconstructor{\Gamma}{\host}{{t=s}{}{}}{X}$}
  \DisplayProof \\
  (trans)\AxiomC{$\termconstructor{\Gamma}{\host}{s=t}{X}$}
  \AxiomC{$\termconstructor{\Gamma}{\host}{t=u}{X}$}
  \BinaryInfC{$\termconstructor{\Gamma}{\host}{{s=u}{}{}}{X}$}
  \DisplayProof \\
  \end{tabular}
  \caption{Evaluation rules for the host language $\host$}\label{hostevalrules}
  \end{figure}
  
  For the $\eta$ rule one has the usual constraints that the variable $x$ does not appear free in the term $s$. 

\section{Enriched categories}\label{sec:appedix}

We recall here some basic background about monoidal categories and enriched categories. See~\cite{BCECT,A2CC} for a complete account.
%%%%%%%%%%%%%%%%%%%%%Monoidal Categories%%%%%%%%%%%%%%%%%%%%%
\subsection{Monoidal categories}\label{sec:monoidal}
It is often useful to reason in a very abstract sense about processes and how they compose. Category theory provides the tool to do this.

A monoidal category is a category equipped with extra data, describing how objects and morphisms can be combined \emph{in parallel}. The main idea is that we can interpret objects of categories as systems, and morphisms as processes. 

One could interpret this for example, as running computer algorithms in parallel, or from a proof-theoretical point of view, as using separate proofs of $P$ and $Q$ to construct a proof of the conjunction ($P$ and $Q$).

\begin{defi}
A \bemph{monoidal category} $\mV=(\mV_0,\otimes,I,a,l,r)$ consists in giving:
\begin{itemize}
\item a category $\mV_0$;
\item an object $I$ of $\mV_0$, called \bemph{the unit};
\item a bifunctor $\freccia{\mV_0\times \mV_0}{\ox}{\mV_0}$, called \bemph{tensor product}, and we write $A\ox B$ for the image under $\ox$ of the pair $(A,B)$;
\item for every $A,B,C$ objects of $\mV_0$, an \bemph{associativity isomorphism}:
\[ \freccia{(A\ox B)\ox C}{a_{ABC}}{A\ox(B\ox C)}\]
such that $\freccia{((-\ox -)\ox -)}{a}{(-\ox(-\ox -))}$ is a natural isomorphism.
\item for every object $A$, a \bemph{left unit} isomorphism
\[ \freccia{I\ox A}{l_A}{A}\]
such that $\freccia{(I\ox -)}{l}{\id_{\mV_0}}$ is a natural isomorphism;
\item for every object $A$, a \bemph{right unit} isomorphism
\[ \freccia{A\ox I}{r_A}{A}\]
such that $\freccia{(-\ox I)}{r}{\id_{\mV_0}}$ is a natural isomorphism.
\end{itemize}
This data must satisfy the \bemph{pentagon} and \bemph{triangle} equations, for all objects $A,B,C$ and $D$ :
\begin{equation}
 \xymatrix{
((A\ox B)\ox C)\ox D \ar[d]_{a_{ABC\ox \id_D}} \ar[rr]^{a_{(A\ox B)CD}} && (A\ox B)\ox( C\ox D) \ar[dd]^{a_{AB(C\ox B)}}\\
(A\ox(B\ox C))\ox D \ar[d]_{a_{A(B\ox C)D}}\\
A\ox ((B\ox C)\ox D)\ar[rr]_{\id_A\ox a_{BCD}} && A\ox(B\ox(C\ox D))
}
\end{equation}

\begin{equation}
\xymatrix{
(A\ox I)\ox B)\ar[drr]_{r_A\ox \id_B} \ar[rr]^{a_{AIB}} && A\ox (I\ox B)\ar[d]^{\id_A\ox l_B}\\
&& A\ox B
}
\end{equation}
A special kind of example, called a cartesian monoidal category, is given by taking for $\mV_0$ any category with finite products, by taking for $\otimes$ and $I$ the product $\times$ and the terminal
object $1$, and by taking for $a$, $l$, $r$ the canonical isomorphisms. 
\end{defi}

Important particular cases of this are the categories $\Set$, $\Cat$, $\Grp$, $\Ord$, $\Top$ of sets, (small) categories, groupoids, ordered sets, topological spaces.

A collection of non-cartesian examples are $\Ab$, $\Hilb$, $\Rel$ of abelian groups, Hilbert spaces, sets and relations.

\begin{defi}
A monoidal category $\mV$ is said to be \bemph{symmetric} when for every $A,B$ there is an isomorphism 
\[\freccia{A\ox B}{s_{AB}}{B\ox A}\]
such that
\begin{itemize}
\item the morphisms $s_{AB}$ are natural in $A,B$;
\item \bemph{associativity coherence}: for every $A,B,C$ the following diagram commutes:
\[\xymatrix{
(A\ox B)\ox  C \ar[rr]^{s_{AB}\ox \id_C} \ar[d]_{a_{ABC}} &&(B\ox A)\ox C \ar[d]^{a_{BAC}}\\
A\ox(B\ox C) \ar[d]_{s_{A(B\ox C)}} && B\ox(A\ox C) \ar[d]^{\id_B\ox s_{AC}}\\
(B\ox C)\ox A \ar[rr]_{a_{BCA}} &&  B\ox(C\ox A)
}\]
\item \bemph{unit coherence}: for every $A$ the following diagram commutes:
\[\xymatrix{
A\ox I\ar[drr]_{r_A} \ar[rr]^{s_{AI}} && I\ox A \ar[d]^{l_A}\\
&& A
}\]
\item \bemph{symmetric axiom}: for every $A,B$ the following diagram commutes:
\[ \xymatrix{
A\ox B \ar[drr]_{\id_{A\ox B}}\ar[rr]^{s_{AB}} && B\ox A \ar[d]^{s_{BA}}\\
&& A\ox B.
}\]
\end{itemize}
\end{defi}
 %%%%per ora no :-)
%\blue{aggiungere monoidal functor e qualche esempio o calcolo grafico?}

\subsection{Enriched categories}\label{sec:background}

In this section we provide some {basic} notions about  enriched categories. For details see. e.g.  \cite{ASGLMC,BCECT,TGTC}. 
%Oversimplifying, an enriched category generalizes the idea of category by replacing \emph{hom-sets}  with objects in another, arbitrary category $V$ (typically, a general monoidal category).

An enriched category is a category in which the hom-functors take their values not in $\Set$, but in some other category $\mV$. 
The theory of enriched categories is now very well developed in category theory, see \cite{BCECT} and \cite{A2CC}, and recently it finds interesting applications in theoretical computer science, see \cite{CCQCLLECT} and \cite{En-lambda}.
For the rest of this section, let $\mV$ be a fixed monoidal category $\mV=(\mV_0,\otimes,I,a,l,r)$, where $\mV_0$ is a category, $\otimes$ is the tensor product, $I$ is the \emph{unit} object of $\mV_0$, $a$ defines the associativity isomorphism and $l$ and $r$ define the left and right unit isomorphism respectively.  

\begin{defi}[Enriched Category]\label{def enriched category}
A $\mV$-\bemph{category} $\mA$ consists of a class $\ob A$ of \bemph{objects}, a \bemph{hom-object} $\mA(A,B)$ of $\mV_0$ for each pair of objects of $\mA$, and
\begin{itemize}
\item \bemph{composition law}
$\freccia{\mA(B,C)\otimes \mA(A,B)}{c_{ABC}}{\mA(A,C)}$
for each triple of objects;
\item \bemph{identity element} 
$\freccia{I}{j_A}{\mA(A,A)}$
for each object subject to
the associativity and unit axioms expressed by the commutativity of the following diagrams

{\scriptsize
\[\xymatrix@+1pc{
(\mA(C,D)\otimes \mA(C,B))\otimes \mA(A,B) \ar[d]^{c_{BCD}\otimes \id}\ar[rr]^{a}&& \mA(C,D)\otimes( \mA(C,B)\otimes \mA(A,B))\ar[d]^{\id\otimes c_{ABC}}\\
\mA(D,B)\otimes \mA(A,B)\ar[r]^{c_{ABD}} &\mA(A,D)& \mA(C,D)\otimes \mA(A,C)\ar[l]^{c_{ACD}}
}\]
\[\xymatrix@+1pc{
\mA(B,B)\otimes \mA(A,B) \ar[r]^{\;\;\;\;\;\;c_{ABB}} & \mA(A,B) & \mA(A,B)\otimes \mA(A,A) \ar[l]^{c_{AAB}\;\;\;\;\;\;}\\
I\otimes \mA(A,B) \ar[u]^{j_B\otimes \id}\ar[ru]^{l}&&\mA(A,B)\otimes I \ar[u]^{\id\otimes j_A} \ar[ul]^{r}.
}\]
}
\end{itemize}
\end{defi}

Taking $\mV=\Set, \Cat, \mathbf{2}, \Ab$ one can re-find the classical notions of (locally small ) ordinary category, 2-category, pre-ordered set, additive category.

\begin{defi}[$\mV$-functor]\label{def:vfunctor}
Let $\mA$ and $\mB$ be $\mV$-categories. A $\mV$-\bemph{functor} $\freccia{\mA}{F}{\mB}$ consists of a function
\[\freccia{\ob{\mA}}{F}{\ob{\mB}}\]
together with, for every pair $A,B\in \ob{\mA}$, a morphism of $\mV$ 
\[\freccia{\mA(A,B)}{F_{AB}}{\mB(FA,FB)}\]
subject to the compatibility with the composition and with the identities expressed by the commutativity of
{\scriptsize
$\xymatrix@+1pc{
\mA(B,C)\otimes \mA(A,B)\ar[d]^{F_{BC}\otimes F_{AB}} \ar[r]^-{c} & \mA(A,C) \ar[d]^{F_{AC}}\\
\mB(FB,FC)\otimes \mB(FA,FB) \ar[r]^{c} & \mB(FA,FC)
}$
}
and
{\scriptsize
$\xymatrix{
& \mA(A,A)\ar[dd]^{F_{AA}}\\
I \ar[ru]^{j_A} \ar[rd]^{j_{FA}}\\
& \mB(FA,FA).
}$
}
\end{defi}

\begin{defi}[$\mV$-Natural Transformations]
Let $\freccia{\mA}{F,G}{\mB}$ be $\mV$-functors. A $\mV$-\bemph{natural transformation} $\freccia{F}{\alpha}{G}$ is an $\ob{\mA}$-indexed family of \bemph{components} 
\[\freccia{I}{\alpha_A}{\mB(FA,GA)}\]
satisfying the $\mV$-naturality condition expressed by the commutativity of the following diagram
{\scriptsize
\[\xymatrix{
 & I \otimes \mA(A,B) \ar[rr]^-{\alpha_B\otimes F_{AB}}&& \mB(FB,GB)\otimes \mB(FA,FB)\ar[dr]^c\\
 \mA(A,B) \ar[dr]^{r^{-1}} \ar[ru]^{l^{-1}} &&&& \mB(FA,GB).\\
 & \mA(A,B)\otimes I \ar[rr]^-{G_{AB}\otimes \alpha_A} && \mB(GA,GB)\otimes \mB(FA,GA) \ar[ru]^c
}
\] 
}

\end{defi}

\end{document}

%% file: macros.tex
\newcommand{\Tadj}[2]{#1^{#2}}
\newcommand{\psalg}[1]{\operatorname{\mathbf{Ps}-#1-\mathbf{Alg}}}
\newcommand{\alglax}[1]{\operatorname{#1\mbox{-}\mathbf{Alg}_l}}
\newcommand{\distlaw}[1]{\operatorname{\mathbf{Dist}_{#1}}}
\newcommand{\pseudodistlaw}[1]{\operatorname{\mathbf{Ps-Dist}_{#1}}}
\newcommand{\lifting}[1]{\operatorname{\mathbf{Lift}}_{\operatorname{#1-\mathbf{Alg}}}}
\newcommand{\pseudolifting}[1]{\operatorname{\mathbf{Lift}}_{\operatorname{\mathbf{Ps}-#1-\mathbf{Alg}}}}
%%%%%%%%%%%%%%%%%Logica Categoriale%%%%%%%%%%%%%%%%%%%%%%%%%%%%%%
\newcommand{\bemph}[1]{\textbf{\emph{#1}}}
\newcommand{\theory}[1]{\mathsf{#1}}
\newcommand{\algtheory}[1]{\mathcal{#1}\mbox{-algebra}}
\newcommand{\context}[3]{[#1_{1}:#2_{1}, \dots ,#1_{#3}:#2_{#3}]}
\newcommand{\termincontext}[3]{#1 : #2\;\; [ #3 ]}
\newcommand{\interptermincontext}[3]{\interp{#1 : #2\;\; [ #3 ]}}
\newcommand{\abbinterptermincontext}[2]{\interp{#1\; [ #2 ]}}
\newcommand{\functionsymbol}[4]{\freccia{#1_1,\dots , #1_{#4}}{#2}{#3}}
\newcommand{\relationsymbol}[3]{#1 \rightarrowtail #2_1,\dots , #2_{#3}}
\newcommand{\interpfunctionsymbol}[4]{\freccia{\interp{#1_1}\times \dots \times \interp{ #1_{#4}}}{\interp{#2}}{\interp{#3}}}
\newcommand{\interpl}{\ensuremath{ [ \hspace{-0.45ex} |}} 
\newcommand{\interpr}{  \ensuremath{ | \hspace{-0.45ex}]}} 
\newcommand{\interp}[1]{\interpl #1 \interpr}
\newcommand{\substitution}[3]{#1[#2/#3]}
\newcommand{\propincontext}[2]{#1 \; prop\; [ #2 ]}
\newcommand{\formulaincontext}[2]{#1 \; [ #2 ]}
\newcommand{\sequentincontext}[3]{#1\vdash #2 \;[#3]}
\newcommand{\sequent}[2]{#1\vdash #2}

\newcommand{\contest}[1]{\vec{#1}}
\newcommand{\Strutture}[1]{\Sigma 	\text{-}\mathbf{Str}(#1)}
\newcommand{\Tmodel}[1]{\cT \text{-}\mathbf{Mod}(#1)}

%%%%% Accenti e overline%%%%%%%%%%%%%%%%%%%%%%%
\newcommand{\ltil}[1]{\widetilde{#1}}
\newcommand{\wht}[1]{\widehat{#1}}
\newcommand{\ovln}[1]{\overline{#1}}
\newcommand{\ovarrow}[1]{\overrightarrow{#1}}
\newcommand{\uln}[1]{\underline{#1}}
%%%%%%%%%%%%%%%%%%%%%%%%%%%%%%%%%%%%%%%%%%%%%%%
\newcommand{\quotientcategory}[1]{{\mQ}_{#1}}                        %categoria di base completata quozienti compeltamento per quozienti
\newcommand{\quotientdoctrine}[1]{#1_q}                          %dotrina completata per quozienti
\newcommand{\quotientcompletion}[1]{\doctrine{\mQ_{#1}}{#1_q}}   %completamento per quozienti

\newcommand{\comprehensioncompletion}[1]{\doctrine{\mG_{#1}}{#1_c}}%completamento per comprensioni
\newcommand{\comprehensioncategory}[1]{{\mG}_{#1}}                    %categoria di base completata comp   
\newcommand{\comprehensiondoctrine}[1]{#1_c}                    %dottrina completata per comprensioni   
\newcommand{\compdiagcompletion}[1]{\doctrine{\mX_{#1}}{{#1}_d}} %completamento diag comprensivi
\newcommand{\compdiagdoctrine}[1]{#1_d}                     %dotrina completata diag comp
\newcommand{\compdiagcategory}[1]{\mX_{#1}}           %categoria di base completata diag comp

\newcommand{\Ef}[1]{\mathbf{Ef}_{#1}}                %categoria regolare costruita da tottrina m-variazionale esistenziale
\newcommand{\compmvardoctrine}[1]{(#1)_{cd}}        %m-variational completion doctrine
\newcommand{\compmvarquotinetdoctrine}[1]{(#1)_{cqd}} 
\newcommand{\tripostotopos}[1]{\mathcal{T}_{#1}}
%%%%%%%%%%%%%%%%%%%%%%%%%%%%%%%%%%%%%%%%%%%%%%%%%%%%%%%%%%%%%%%%%%%%%%%%%%
\newcommand{\compex}[1]{{#1}^{\existential}}
\newcommand{\compel}[1]{{#1}^{\elementary}}
\newcommand{\angbr}[2]{\langle #1,#2 \rangle} 
\newcommand{\comprl}{\ensuremath{ \{ \hspace{-0.6ex} |}} 
\newcommand{\comprr}{  \ensuremath{ | \hspace{-0.6ex}\}}} 
\newcommand{\comp}[1]{\comprl #1 \comprr}
\newcommand{\freccia}[3]{\xymatrix{#2 \colon #1  \ar[r] &  #3}}
\newcommand{\cover}[3]{\xymatrix{#2 \colon #1  \ar@{-|>}[r] &  #3}}
\newcommand{\frecciasopra}[3]{\xymatrix{ #1  \ar[r]^{#2} &  #3}}
\newcommand{\frecciasopralunga}[3]{\xymatrix{ #1  \ar[rr]^{#2} &&  #3}}
\newcommand{\pbmorph}[2]{#1^{\ast}#2} 
\newcommand{\duefreccia}[3]{\xymatrix@C=0.5cm{#2 \colon #1  \ar@{=>}[r] &  #3}}
\newcommand{\modificazione}[3]{\xymatrix@C=0.5cm{#2 \colon #1  \ar@{~>}[r] &  #3}}
\newcommand{\doctrine}[2]{\xymatrix{#2 \colon #1^{\op}  \ar[r] & \infsl }}
\newcommand{\hyperdoctrine}[2]{\xymatrix{#2 \colon #1^{op}  \ar[r] & \heyalg }}
\newcommand{\equivalence}[3]{#2 \colon #1 \equiv #3}
\newcommand{\duemorfismo}[6]{\xymatrix{
#1^{\op} \ar[rrd]^#2_{}="a" \ar[dd]_{#3^{\op}}\\
&& \infsl\\
#5^{\op}  \ar[rru]_#6^{}="b"
\ar_#4  "a";"b"}}
\newcommand{\comsquare}[8]{ \xymatrix@+1pc{ 
#1 \ar[r]^{#5} \ar[d]_{#6} & #2 \ar[d]^{#7} \\
#3 \ar[r]_{#8} & #4 
}}
\newcommand{\pullback}[8]{ \xymatrix@+1pc{ 
#1 \pullbackcorner \ar[r]^{#5} \ar[d]_{#6} & #2 \ar[d]^{#7} \\
#3 \ar[r]_{#8} & #4 
}}
\newcommand{\quadratocomm}[8]{ \xymatrix@+1pc{ 
#1 \ar[r]^{#5} \ar[d]_{#6} & #2 \ar[d]^{#7} \\
#3 \ar[r]_{#8} & #4 
}}
\newcommand{\comsquarelargo}[8]{ \xymatrix@+1pc{ 
#1 \ar[rr]^{#5} \ar[d]_{#6} && #2 \ar[d]^{#7} \\
#3 \ar[rr]_{#8} && #4 
}}
\newcommand{\parallelmorphisms}[4]{\xymatrix@+1pc{
#1 \ar @<+4pt>[r]^{#2} \ar @<-4pt>[r]_{#3} & #4
}}
\newcommand{\relation}[4]{\xymatrix@+1pc{
\angbr{#2}{#3}\colon #1 \ar @<+4pt>[r] \ar @<-4pt>[r] & #4
}}
\newcommand{\frecceparalleleopposte}[4]{\xymatrix@+1pc{
#1 \ar@<+4pt>[r]^{#2} \ar@<-4pt>@{<-}[r]_{#3} & #4
}}
\newcommand{\equalizer}[6]{\xymatrix@+1pc{
#1 \ar[r]^{#2} & #3 \ar @<+4pt>[r]^{#4} \ar @<-4pt>[r]_{#5} & #6
}}
\newcommand{\coequalizer}[6]{\xymatrix@+1pc{
 #1 \ar @<+4pt>[r]^{#2} \ar @<-4pt>[r]_{#3} & #4 \ar[r]^{#5} & #6
}}
\newcommand{\sottoggetto}[2]{\xymatrix{
#1 \ar@{>->}[r] & #2
}}
\newcommand{\subobject}[3]{\xymatrix{
#1 \ar@{>->}[r]^{#2} & #3
}}
\newcommand{\spanbelongsrelation}[3]{#2\left(#3\right)#1}
\newcommand{\RFinverso}[1]{R_{\langle F(\pr_1),F(\pr_2)\rangle^{-1}}(#1)}
\newcommand{\RFnormale}[1]{R_{\langle F(\pr_1),F(\pr_2)\rangle}(#1)}
\newcommand{\RGinverso}[1]{R_{\langle G(\pr_1),G(\pr_2)\rangle^{-1}}(#1)}
\newcommand{\pullbackcorner}[1][ul]{\save*!/#1+1.2pc/#1:(1,-1)@^{|-}\restore}
%%%%%%%%%%%%%%%%%%%%%%%%%%%%%%%%%%%%%%%%%%%%%%%%%%%
\def\Heyt{\operatorname{\mathbf{Heyt}}}
\def\Gp{\operatorname{\mathbf{Gp}}}
\def\Grp{\operatorname{\mathbf{Grp}}}
\def\Ab{\operatorname{\mathbf{Ab}}}
\def\Hilb{\operatorname{\mathbf{Hilb}}}
\def\Rel{\operatorname{\mathbf{Rel}}}

\def\Mon{\operatorname{\mathbf{Mon}}}
\def\Top{\operatorname{\mathbf{Top}}}
\def\Ord{\operatorname{\mathbf{Ord}}}

\def\Fp{\operatorname{\mathbf{Fp}}}

\def\FinComp{\operatorname{\mathbf{FinComp}}}
\def\ED{\operatorname{\mathbf{ExD}}}            %dot. esistenziali
\def\ElD{\operatorname{\mathbf{ElD}}}         %dot. elementari
\def\PD{\operatorname{\mathbf{PD}}}             %dot. primarie
\def\PDD{\operatorname{\mathbf{PdD}}}           %dot. primarie con base discreta
\def\EDD{\operatorname{\mathbf{EdD}}}        %dot. esistenziali con base discreta 
\def\CEED{\operatorname{\mathbf{Ex-mVar}}}       %dottrine esistenziali m-variazionali
\def\PED{\operatorname{\mathbf{ExD}}}   %dottrine esistenziali
\def\Cat{\operatorname{\mathbf{Cat}}}
\def\EED{\operatorname{\mathbf{EED}}}           %dottrine elementari esistenziali
\def\SD{\operatorname{\mathbf{Ex-mVar}}}           %dottrine esistenziali m-variazionali
\def\mVar{\operatorname{\mathbf{mVar}}}           %dottrine esistenziali m-variazionali
\def\CED{\operatorname{\mathbf{CED}}}         %dot. elementari con diagonali comprensioni
\def\CE{\operatorname{\mathbf{CE}}}           %dot. elementari con comprensioni
\def\Reg{\operatorname{\mathbf{Reg}}}
\def\QD{\operatorname{\mathbf{QD}}}           %dot. m-variazionali con quozienti stabili ed effettivi
\def\QED{\operatorname{\mathbf{QED}}}          %categoria dottrine elementari con quozienti
\def\LFS{\operatorname{\mathbf{LFS}}}          %categoria sistemi di fattorizzazione stabili
\def\excat{\operatorname{\mathbf{Xct}}}        %categore exatte
\def\mR{\mathbb{\mathcal{R}}}
\newcommand{\rel}[1]{\mathbf{Rel}(#1)}
\newcommand{\erel}[1]{\mathbf{ERel}(#1)}
\newcommand{\subobcategory}[1]{\mathbf{Sub}(#1)}
\newcommand{\maprel}[1]{\mathbf{Map \; Rel}(#1)}
%%%%%%%%%%%%%%%%%%%%%%%% Categorie %%%%%%%%%
\def\mA{\mathcal{A}}
\def\mB{\mathcal{B}}
\def\mC{\mathcal{C}}
\def\mX{\mathcal{X}}
\def\mD{\mathcal{D}}
\def\mE{\mathcal{E}}
\def\mF{\mathcal{F}}
\def\mM{\mathcal{M}}
\def\mK{\mathcal{K}}
\def\mO{\mathcal{O}}
\def\mQ{\mathcal{Q}}
\def\mV{\mathcal{H}}
\def\terminalcat{\mathcal{I}}
%%%%%%%%%%%%%%%%%%%%%%%%%%%%%%%%%%%%%%%%%%%%

\def\mG{\mathcal{G}}
\def\cT{\mathbb{T}}    %usato per indicare teorie
\def\forg{\mathrm{U}}
\def\des{Des}         %desent data
%%%%%%%%%%%%%%%%%%% Funtori %%%%%%%%%%%%%%%%
\def\compfun{\mathrm{c}}
\def\unitfun{\mathrm{u}}
\def\funD{\mathrm{D}}
\def\funQ{\mathrm{Q}}
\def\funC{\mathrm{C}}
\def\funG{\mathrm{G}}
\def\funE{\mathrm{E}}
\def\funS{\mathrm{S}}
\def\funF{\mathrm{F}}
\def\funEl{\mathrm{El}}
\def\funU{\mathrm{U}}
\def\funT{\mathrm{T}}
\def\funK{\mathrm{K}}
\def\mT{\mathrm{T}}
\def\mS{\mathrm{S}}
\def\terminalobject{1}
%%%%%%%%%%%%%%%%%%%%%%%%%%%%%%%%%%%%%%%%%%%%%
\def\LindTar{\doctrine{\mathcal{V}}{LT}}
\def\ox{\otimes}
\newbox\erove \setbox\erove=\hbox{\reflectbox{E}}
\def\Einv{\usebox\erove}
\def\pr{\operatorname{ pr}}
\def\FV{\operatorname{ FV}}
\def\id{\operatorname{ id}}
\def\iso{\operatorname{ Iso}}
\def\op{\operatorname{ op}}
\def\cod{\operatorname{ cod}}
\def\dom{\operatorname{ dom}}
\def\im{\operatorname{ im}}
\def\mono{\operatorname{ mono}}
\def\el{\operatorname{ el}}
\def\existential{\operatorname{ ex}}
\def\elementary{\operatorname{ el}}
\def\theory{\operatorname{ \mathsf{T}}}
\def\Sub{\operatorname{ Sub}}

\def\Sh{\operatorname{ \mathbf{Shv}}}
%%%%%% T-algebra e pseudo-T-algebra%%%%%%%%%%%%%%%
\def\Talg{\operatorname{ T-\mathbf{Alg}}}
\def\Txalg{\operatorname{ T_{x}-\mathbf{Alg}}}
\def\Talglax{\operatorname{ T-\mathbf{Alg}_l}}
\def\Tcalglax{\operatorname{ T_c-\mathbf{Alg}_l}}
\def\PseudoTalg{\operatorname{ \mathbf{Ps}-T-\mathbf{Alg}}}
%%%%%%%%%%%%%%%%%%%%%%%%%%%%%%%%%%%%%%%%%%%%%%%%%%
\newcommand{\exactcomp}[1]{(#1)_{\ex / \reg}}      %exact completion of a regular cat
\newcommand{\regularcomp}[1]{(#1)_{\reg / \lex}}   %regular completion of a lex-cat
\newcommand{\exactcomplex}[1]{(#1)_{\ex / \lex}}   %exact completion of a lex-cat
\def\reg{\operatorname{reg}}
\def\ex{\operatorname{ex}}
\def\lex{\operatorname{lex}}
\def\Set{\operatorname{\mathbf{Set}}}
\def\Sh{\operatorname{\mathbf{Sh}}}
\def\yoneda{\operatorname{\mathbf{Y}}}
\def\signature{\operatorname{\mathbf{Sg}}}
\def\terms{\operatorname{\mathbf{Terms}}}
\def\variables{\operatorname{\mathbf{Var}}}
\def\:{\colon}
\def\infsl{\operatorname{\mathbf{InfSL}}}
\def\heyalg{\operatorname{\mathbf{HeyAlg}}}
\def\pos{\operatorname{\mathbf{Pos}_{\top}}}
\def\ev{\operatorname{ev}}
\def\equalsymb{\operatorename{Eq}}
\def\mfa{\mathfrak{a}}
\def\LT{\mathcal{L}}
%%%%%%%%%%%%%%%2-cell composition%%%%%%%%%%%%%%%%%%%%%%%%%%%%%%%%%%%%%
\def\parallelcomp{.}
\def\verticalcomp{\circ}
\def\2cellidentity{i}
\def\1cellidentity{1}
\def\cellcomp{\circ} 

%%%%%%%%%%%%%%%%%%%%%%%%%%%%%%%%%%%%%%%%%%%%%%%%%
\def\type{\operatorname{\mathbf{type}}}
\newcommand{\ob}[1]{\mathbf{ob}(#1)}
\newcommand{\typeconstructor}[2]{\vdash_{#1} #2 \; : \type}
\newcommand{\termconstructor}[4]{#1 \vdash_{#2} #3\;: #4}
\newcommand{\terconstrmix}[5]{#1\;|\; #2 \vdash_{#3} #4\;: #5}
\newcommand{\inr}[1]{\pi_2(#1)}
\newcommand{\inl}[1]{\pi_1(#1)}
\newcommand{\letbein}[3]{\mathbf{let} \; #1\; \mathbf{be} \; #2 \;  \mathbf{in}\; #3}
\newcommand{\Ctype}[3]{#1(#2,#3)}
\newcommand{\arrowtype}[2]{#1\rightarrow #2}
\newcommand{\oput}[1]{\mathbf{output}(#1)}
\newcommand{\derelict}[1]{\mathbf{derelict}(#1)}
\newcommand{\promote}[1]{\mathbf{promote}(#1)}
\newcommand{\enrichedtheory}[2]{\mathcal{T}(#1,#2)}
\newcommand{\Carrowtype}[2]{\xymatrix@-1pc{#1 \ar@{-o}[r] & #2}}
\newcommand{\Vcat}[1]{#1 \mbox{-} \Cat}
\newcommand{\enrichedcomp}[2]{\mathbf{comp}(#1,#2)}
\def\VCcalcolo{\operatorname{\mathsf{H\hspace{-0.28ex}C}}}
\def\Smc{\operatorname{\mathbf{SMC}}}
\def\Pll{\operatorname{\mathbf{PLL}}}
\def\cstar{\operatorname{C^{\ast}}}
\def\wstar{\operatorname{W^{\ast}}}
\def\fdcstar{\operatorname{\mathbf{FdC^{\ast}-Alg}}}
\def\fdwstar{\operatorname{\mathbf{FdW^{\ast}-Alg}}}
\def\dcpo{\operatorname{\mathbf{Dcpo_{\bot}}}}
\def\catvcalcolo{\mathrm{Th}(\VCcalcolo)}
\newcommand{\vcat}[1]{#1\mbox{-}\Cat}
\newcommand{\modelcat}[1]{\mathrm{Model}(#1)}

%% file: BackupLongVersion.bbl
\begin{thebibliography}{MMdPR05}

\bibitem[BCM20]{BuroCM20}
Samuele Buro, Roy~L. Crole, and Isabella Mastroeni.
\newblock Equational logic and categorical semantics for multi-languages.
\newblock In Patricia Johann, editor, {\em Proceedings of the 36th Conference
  on the Mathematical Foundations of Programming Semantics, {MFPS} 2020,
  Online, October 1, 2020}, volume 352 of {\em Electronic Notes in Theoretical
  Computer Science}, pages 79--103. Elsevier, 2020.

\bibitem[Ben95]{AMLNLLPTM}
P.~N. Benton.
\newblock A mixed linear and non-linear logic: Proofs, terms and models.
\newblock In Leszek Pacholski and Jerzy Tiuryn, editors, {\em Computer Science
  Logic}, pages 121--135, Berlin, Heidelberg, 1995. Springer Berlin Heidelberg.

\bibitem[BM19]{BuroM19}
Samuele Buro and Isabella Mastroeni.
\newblock On the multi-language construction.
\newblock In Lu{\'{\i}}s Caires, editor, {\em Programming Languages and Systems
  - 28th European Symposium on Programming, {ESOP} 2019, Held as Part of the
  European Joint Conferences on Theory and Practice of Software, {ETAPS} 2019,
  Prague, Czech Republic, April 6-11, 2019, Proceedings}, volume 11423 of {\em
  Lecture Notes in Computer Science}, pages 293--321. Springer, 2019.

\bibitem[BW90]{CTCS}
M.~Barr and C.~Wells.
\newblock {\em Category Theory for Computing Science}.
\newblock Prentice-Hall, Inc., USA, 1990.

\bibitem[CdVW19]{CVW19}
Pierre Clairambault, Marc de~Visme, and Glynn Winskel.
\newblock Game semantics for quantum programming.
\newblock {\em Proc. {ACM} Program. Lang.}, 3({POPL}):32:1--32:29, 2019.

\bibitem[Chi13]{Chisnall13}
David Chisnall.
\newblock The challenge of cross-language interoperability.
\newblock {\em Commun. {ACM}}, 56(12):50--56, 2013.

\bibitem[EMS14]{EggerMS14}
Jeff Egger, Rasmus~Ejlers M{\o}gelberg, and Alex Simpson.
\newblock The enriched effect calculus: syntax and semantics.
\newblock {\em J. Log. Comput.}, 24(3):615--654, 2014.

\bibitem[GL87]{LLLC}
J.~Y. Girard and Y.~Lafont.
\newblock Linear logic and lazy computation.
\newblock In Hartmut Ehrig, Robert Kowalski, Giorgio Levi, and Ugo Montanari,
  editors, {\em TAPSOFT '87}, pages 52--66, Berlin, Heidelberg, 1987. Springer
  Berlin Heidelberg.

\bibitem[Joh02]{SAE}
P.T. Johnstone.
\newblock {\em Sketches of an elephant: a topos theory compendium}, volume~2 of
  {\em Studies in Logic and the foundations of mathematics}.
\newblock Oxford Univ. Press, 2002.

\bibitem[Js91]{TGTC}
A.~Joyal and R.~street.
\newblock The geometry of tensor calculus i.
\newblock {\em Advances in Mathe-matics}, 88(1):55--112, 1991.

\bibitem[Kel05]{BCECT}
G.M. Kelly.
\newblock Basic concepts of enriched category theory.
\newblock {\em Theory Appl. Categ.}, 10, 2005.

\bibitem[Lac10]{A2CC}
S~Lack.
\newblock {\em A 2-Categories Companion}, pages 105--191.
\newblock Springer New York, New York, NY, 2010.

\bibitem[LMZ18]{En-lambda}
Bert Lindenhovius, Michael~W. Mislove, and Vladimir Zamdzhiev.
\newblock Enriching a linear/non-linear lambda calculus: {A} programming
  language for string diagrams.
\newblock In Anuj Dawar and Erich Gr{\"{a}}del, editors, {\em Proceedings of
  the 33rd Annual {ACM/IEEE} Symposium on Logic in Computer Science, {LICS}
  2018, Oxford, UK, July 09-12, 2018}, pages 659--668. {ACM}, 2018.

\bibitem[LS86]{IHOCL}
J.~Lambek and P.~J. Scott.
\newblock {\em Introduction to Higher Order Categorical Logic}.
\newblock Cambridge Univ. Press, 1986.

\bibitem[MdPR00]{milly2}
Maria~Emilia Maietti, Valeria de~Paiva, and Eike Ritter.
\newblock Categorical models for intuitionistic and linear type theory.
\newblock In Jerzy Tiuryn, editor, {\em Foundations of Software Science and
  Computation Structures}, pages 223--237, Berlin, Heidelberg, 2000. Springer
  Berlin Heidelberg.

\bibitem[MF07]{MatthewsF07}
Jacob Matthews and Robert~Bruce Findler.
\newblock Operational semantics for multi-language programs.
\newblock In Martin Hofmann and Matthias Felleisen, editors, {\em Proceedings
  of the 34th {ACM} {SIGPLAN-SIGACT} Symposium on Principles of Programming
  Languages, {POPL} 2007, Nice, France, January 17-19, 2007}, pages 3--10.
  {ACM}, 2007.

\bibitem[MMdPR05]{RCSILL}
M.~E. Maietti, P.~Maneggia, V.~de~Paiva, and E.~Ritter.
\newblock Relating categorical semantics for intuitionistic linear logic.
\newblock {\em Appl. Categ. Structures}, 13(1):1--36, 2005.

\bibitem[OSZ12]{OseraSZ12}
Peter{-}Michael Osera, Vilhelm Sj{\"{o}}berg, and Steve Zdancewic.
\newblock Dependent interoperability.
\newblock In Koen Claessen and Nikhil Swamy, editors, {\em Proceedings of the
  sixth workshop on Programming Languages meets Program Verification, {PLPV}
  2012, Philadelphia, PA, USA, January 24, 2012}, pages 3--14. {ACM}, 2012.

\bibitem[Pay18]{Paykin18}
Jennifer Paykin.
\newblock {\em Linear/Non-Linear Types for Embedded Domain-Specific Languages}.
\newblock PhD thesis, Ph.D. Thesis, University of Pennsyvlania, 2018.

\bibitem[Pit95]{CLP}
A.~M. Pitts.
\newblock Categorical logic.
\newblock In S.~Abramsky, D.~M. Gabbay, and T.~S.~E. Maibaum, editors, {\em
  Handbook of Logic in Computer Science}, volume~6, pages 39--.129. Oxford
  Univ. Press, 1995.

\bibitem[PPDA17]{PattersonPDA17}
Daniel Patterson, Jamie Perconti, Christos Dimoulas, and Amal Ahmed.
\newblock Funtal: reasonably mixing a functional language with assembly.
\newblock In Albert Cohen and Martin~T. Vechev, editors, {\em Proceedings of
  the 38th {ACM} {SIGPLAN} Conference on Programming Language Design and
  Implementation, {PLDI} 2017, Barcelona, Spain, June 18-23, 2017}, pages
  495--509. {ACM}, 2017.

\bibitem[PPZ19]{PPZ19}
Luca Paolini, Mauro Piccolo, and Margherita Zorzi.
\newblock {QPCF:} higher-order languages and quantum circuits.
\newblock {\em Journal Automated Reasoning}, 63(4):941--966, 2019.

\bibitem[PRZ17]{qwire17a}
Jennifer Paykin, Robert Rand, and Steve Zdancewic.
\newblock Qwire: A core language for quantum circuits.
\newblock In {\em Proceedings of the 44th ACM SIGPLAN Symposium on Principles
  of Programming Languages}, POPL 2017, pages 846--858, New York, NY, USA,
  2017. ACM.

\bibitem[Ran18]{RobRandThesis}
Robert Rand.
\newblock {\em Formally Verified Quantum Programming PhD Thesis}.
\newblock PhD thesis, University of Pennsylvania, USA, 2018.
\newblock http://www.cis.upenn.edu/~rrand/thesis.pdf.

\bibitem[Rie14]{CHT}
E.~Riehl.
\newblock {\em Categorical homotopy theory}.
\newblock Cambridge University Press, 2014.

\bibitem[RPLZ19]{qwire3}
R.~Rand, J.~Paykin, D.-H. Lee, and S.~Zdancewic.
\newblock Reqwire: Reasoning about reversible quantum circuits.
\newblock In {\em 15th International Conference on Quantum Physics and Logic,
  QPL 2018; Dalhousie UniversityHalifax; Canada; 3 June 2018 through 7 June
  2018}, volume 287, pages 299--312, 2019.

\bibitem[RPZ17]{qwire17b}
Robert Rand, Jennifer Paykin, and Steve Zdancewic.
\newblock Qwire practice: Formal verification of quantum circuits in coq.
\newblock In {\em Proceedings 14th International Conference on Quantum Physics
  and Logic, Nijmegen, The Netherlands, 3-7 July 2017}, volume 266 of {\em
  EPTCS}, pages 119--132, 2017.

\bibitem[RS17]{CCQCLLECT}
M.~Rennela and S.~Staton.
\newblock Classical control, quantum circuits and linear logic in enriched
  category theory.
\newblock {\em CoRR}, abs/1711.05159, 2017.

\bibitem[RS18]{RENNELA2018257}
Mathys Rennela and Sam Staton.
\newblock Classical control and quantum circuits in enriched category theory.
\newblock {\em Electronic Notes in Theoretical Computer Science}, 336:257 --
  279, 2018.
\newblock The Thirty-third Conference on the Mathematical Foundations of
  Programming Semantics (MFPS XXXIII).

\bibitem[RS20a]{RS20}
Mathys Rennela and Sam Staton.
\newblock Classical control, quantum circuits and linear logic in enriched
  category theory.
\newblock {\em Log. Methods Comput. Sci.}, 16(1), 2020.

\bibitem[RS20b]{RennelaS19}
Mathys Rennela and Sam Staton.
\newblock Classical control, quantum circuits and linear logic in enriched
  category theory.
\newblock {\em Log. Methods Comput. Sci.}, 16(1), 2020.

\bibitem[Sak12]{CAWA}
S.~Sakai.
\newblock {\em C*-Algebras and W*-Algebras}.
\newblock Classics in Mathematics. Springer Berlin Heidelberg, 2012.

\bibitem[Sel07]{Selinger07a}
Peter Selinger.
\newblock Dagger compact closed categories and completely positive maps:
  (extended abstract).
\newblock {\em Electron. Notes Theor. Comput. Sci.}, 170:139--163, 2007.

\bibitem[Sel09]{ASGLMC}
P.~Selinger.
\newblock A survey of graphical languages for monoidal categories.
\newblock {\em Lecture Notes in Physics}, 813, 08 2009.

\bibitem[Zor16]{Zorzi16}
Margherita Zorzi.
\newblock On quantum lambda calculi: a foundational perspective.
\newblock {\em Mathematical Structures in Computer Science}, 26(7):1107--1195,
  2016.

\end{thebibliography}
